\documentclass[letterpaper,11pt]{article}
\usepackage{verbatim, url}
\usepackage{latexsym}
\usepackage{amsmath,amssymb,amsthm}
\usepackage{epsfig}
\usepackage{fullpage}
\usepackage{times}

\usepackage{tikz, fp, subfigure, epsfig}
\usepackage{hyperref, cite}
\usepackage{color}

\usepackage{verbatim, url}
\usepackage{latexsym}
\usepackage{amsmath,amssymb,amsthm}
\usepackage{epsfig}
\usepackage{fullpage}
\usepackage{times}
\usepackage{tikz, fp, subfigure, epsfig, filecontents, float}
\usepackage{hyperref, cite}
\usepackage{color}
\usepackage{authblk}

\newcommand\drop[1]{}
\newcommand\mtcom[1]{\begin{quote}{\color{red} Mikkel: {#1}}\end{quote}}

\let\myPushQED=\pushQED
\let\myPopQED=\popQED
\newcommand{\myignore}[1]{}
\newenvironment{proof*}
  {\let\pushQED=\myignore\begin{proof}\let\pushQED=\myPushQED}
  {\def\popQED{}\end{proof}\let\popQED=\myPopQED}

\newenvironment{itemize*}%
  {\vspace{-1ex}\begin{itemize}%
    \setlength{\itemsep}{-0.5ex}%
    \setlength{\parsep}{0pt}}%
  {\end{itemize}}

\newenvironment{enumerate*}%
  {\vspace{-1ex}\begin{enumerate}%
    \setlength{\itemsep}{-0.5ex}%
    \setlength{\parsep}{0pt}}%
  {\end{enumerate}}

\newenvironment{description*}%
  {\vspace{-1ex}\begin{description}%
    \setlength{\itemsep}{-0.5ex}%
    \setlength{\parsep}{0pt}}%
  {\end{description}}

{\makeatletter
 \gdef\xxxmark{%
   \expandafter\ifx\csname @mpargs\endcsname\relax 
     \expandafter\ifx\csname @captype\endcsname\relax 
       \marginpar{xxx}
     \else
       xxx 
     \fi
   \else
     xxx 
   \fi}
 \gdef\xxx{\@ifnextchar[\xxx@lab\xxx@nolab}
 \long\gdef\xxx@lab[#1]#2{{\bf [\xxxmark #2 ---{\sc #1}]}}
 \long\gdef\xxx@nolab#1{{\bf [\xxxmark #1]}}
}

\newtheorem{theorem}{Theorem}
\newtheorem{lemma}[theorem]{Lemma}
\newtheorem{corollary}[theorem]{Corollary}

\newtheorem{invariant}[theorem]{Invariant}
\newtheorem{remark}[theorem]{Remark}

\newcommand{\eps}{\varepsilon}

\newcommand{\req}[1]{(\ref{#1})}

\newcommand{\E}{\mathbf{E}}

\newcommand{\ceil}[1]{\lceil {#1} \rceil}
\newcommand{\floor}[1]{\lfloor {#1} \rfloor}

\let\phi=\varphi

\newcommand{\cD}{{\mathcal D}}

\newcommand{\balls}{m}
\newcommand{\bins}{n}

\renewcommand{\th}{\ifmmode{^{\textrm{th}}}\else{\textsuperscript{th}\ }\fi}

\title{Consistent Hashing with Bounded Loads}
\author[1]{Vahab Mirrokni}
\author[2]{Mikkel Thorup\thanks{Supported in part by
   Advanced Grant DFF-0602-02499B from the Danish Council for
   Independent Research under the Sapere Aude research career programme.}}
\author[1]{Morteza Zadimoghaddam}
\affil[1]{Google Research, New York. \texttt{\{mirrokni,zadim\}@google.com}}
\affil[2]{University of Copenhagen. \texttt{mikkel2thorup@gmail.com}}

\begin{document}
\maketitle

\thispagestyle{empty}

\begin{abstract}
In dynamic load balancing, we wish to allocate 
a set of clients (balls) to a set of servers (bins) with the goal
of minimizing the maximum load of any server and also minimizing  
the number of moves after adding or removing a server
or a client\footnote{Throughout this paper, we use balls
  as clients, and also bins as servers, interchangeably.}.
We want a hashing-style solution where we given the ID of a client
can efficiently find its server in a distributed dynamic environment. 
In such a dynamic environment, 
both servers and clients may be added and/or removed from the
system in any order.
The most popular solutions for such dynamic settings are
Consistent Hashing~\cite{chord-theory,chord} or Rendezvous Hashing \cite{TR98:rendezvous}. However, the load balancing of these
schemes  is no better than a random
assignment of clients to servers, so with $\bins$ of each, we expect
many servers to be overloaded with $\Theta(\log \bins/ \log\log
\bins)$ clients. In this paper, we aim to design hashing schemes
that achieve any desirable level of load balancing, while minimizing
the number of movements under any addition or removal of servers or clients.

In particular, we consider a problem with $\balls$
balls and $\bins$ bins, and given a user-specified balancing parameter $c=1+\eps>1$,  
 we aim to find a hashing scheme with no load above $\ceil{c\balls/\bins}$, referred to
 as the {\em capacity} of the bins. 
Our algorithmic starting point is the consistent hashing scheme 
where current balls and bins are hashed to a
unit cycle, and a ball is placed in the first bin succeeding it in
clock-wise order. In order to cope with given capacity constraints, 
we  apply the idea of {\em linear probing by forwarding the ball on the circle 
to the first non-full bin}.  
We show that in our hashing scheme when a ball or bin is inserted or deleted, the expected number of balls that have to be
moved is within a multiplicative factor  of $O({1\over \eps^2})$  of the optimum for $\eps
\le 1$ (Theorem~\ref{thm:moves-eps}) and within a factor $1+O(\frac{\log c}c)$  of the optimum for
$\eps\ge 1$ (Theorem~\ref{thm:moves-c}). Technically, the latter bound is the most
challenging to prove. It implies that we for superconstant $c$, we only
pay a negligible cost in extra moves.  We also get the same bounds for the
simpler problem where we instead of a user specified balancing
parameter have a fixed bin capacity $C$ for all bins, and define
$c=1+\eps=C\bins/\balls$.
\end{abstract}
\setcounter{page}0
\newpage
\section{Introduction}\label{sec:intro}
Load balancing in dynamic environments is a central problem in
designing several networking systems and web
services~\cite{chord,chord-theory}. We wish to allocate {\em
  clients\/} (also referred to as {\em balls}) to {\em servers\/}
(also referred to as {\em bins}) in such a way that none of the
servers gets overloaded. Here, the {\em load\/} of a server is the
number of clients allocated to it. We want a hashing-style solution
where we given the ID of a client can efficiently find its server.
Both clients and servers
may be added or removed in any order, and with such changes, we do not
want to move too many clients. Thus, while the dynamic allocation
algorithm has to always ensure a proper load balancing, it should aim
to minimize the number of clients moved after each change to the
system.  
For every update in the system, we need to change the allocation of clients to servers. 
For simplicity, we assume that the updates (ball and bin additions and removals) do not happen simultaneously and will be operated one at  a time, so that
we have time to finish changing the allocation before
we get another update.

Note that such allocation problems become even more
challenging when we face hard constraints in the capacity of each
server, that is, each server has a {\em capacity\/} and the load may
not exceed this capacity. Typically, we want capacities close to the
average loads.  

There is a vast literature on solutions in the much
simpler case where the set of servers is fixed and only the client set
is updated, and we shall return to some of these in Section
\ref{sec:power}. For now, we focus on solutions that are known to work
in our fully-dynamic case where both clients and servers can be added
and removed in an arbitrary order.  This also rules out solutions where
only the last added server may be removed\footnote{This rules out the
  external memory techniques \cite{Lar88} where blocks (playing the
  role of fixed capacity servers) can only be added to and removed
  from the top of the current memory.}. The above problem formulation
is very general, and does not assume anything about the ratio between
the number $\balls$ of clients and the number $\bins$ of servers,
e.g., processors are cheap, so one can imagine systems with a large
number of servers, but there could also be many clients. One could
also imagine a balanced system in which $\balls=\bins$. 

The classic solution to the scenario where both clients and servers
can be added and removed is Consistent
Hashing~\cite{chord,chord-theory} where the current clients are
assigned in a random way to the current servers.  While consistent
hashing schemes minimize the expected number of movements, they may
result in hugely overloaded servers, and they do not allow for
explicit capacity constraints on the servers. The basic point is that
the load balancing of consistent hashing~\cite{chord-theory,chord} is
no better than a random assignment of clients to servers. The
same issue holds for Rendezvous or Highest Random Weight
Hashing \cite{TR98:rendezvous}\footnote{The popular name ``Randezvous Hashing'' appears to stem from an earlier technical report, but even
in that report, randezvous is used about a more general situation, and
not about the Highest Random Weigth Hashing that has now become
known popurlarly known as Randezvous hashing.}.
Hence, with
$\balls$ clients and $\bins$ servers, we expect a good load
balancing if $\balls/\bins=\Omega(\log \bins)$, but the balance is
lost with smaller loads, e.g., with $\balls=\bins$, we expect many
servers to be overloaded with $\Theta(\log \bins/ \log\log \bins)$
clients.  In this paper, with $\bins$ clients and $\bins$ servers, we
get a guaranteed max-load of 2, while only moving an expected constant
number of clients each time a client or server is added or
removed. Our general result is described below.

In this paper, we present an algorithm that works with arbitrary
capacity constraints on the servers. For the purpose of load
balancing, the system designer can specify a balancing parameter
$c=1+\eps$, guaranteeing that the maximum load is at most
$\ceil{c\balls/\bins}$.  While maintaining this hard balancing
constraint, we limit the expected number of clients to be moved when
clients or servers are inserted or removed.

Even without capacity constraints, the obvious general lower bounds
for moves are as follows. When a client is added or removed, at least
we have to move that client.  When a server is added or removed, at
least we have to move the clients belonging to it. On the average, we
therefore have to move least ${\balls \over \bins}$ clients when a server is
added or removed.

With our solution, while guaranteeing a balancing parameter $c=1+\eps\leq 2$,
when a client is added or removed, the expected number of
clients moved is $O({1\over \eps^2})$. When a server is added or removed, 
the expected number of clients moved is $O({\balls \over \eps^2\bins })$ 
(Theorem~\ref{thm:moves-eps}). These numbers are only a factor $O({1\over \eps^2})$ 
worse than the general lower bounds without capacity constrains.
For balancing parameter $c\geq 2$, our expected number of moves
is increased by a factor $1+O(\frac{\log c}c)$ over the 
lower bounds (Theorem~\ref{thm:moves-c}). The bound for $c\geq 2$ is the most challenging to prove. It
implies that for superconstant $c$, we only expect to pay a negligible cost in
extra moves.

From a more
practical perspective, it is useful that our algorithm provides a simple knob, the
load balancing parameter $c=1+\eps$, which captures the tradeoff between
load balancing and stability upon changes in
the system. This gives a more direct control to the
system designer in meeting explicit balancing constraints.

We get the same bounds for the
simpler problem where we instead of a user specified balancing
parameter have a fixed bin capacity $C$ for all bins and define
$c=1+\eps=C\bins/\balls$. 
An interesting case of light loads is the unit capacity case where each server
can only handle a single client at the time.

\paragraph{Applications.}
Consistent hashing has found numerous applications~\cite{OV11,GF12}
and early work in this area \cite{chord-theory,chordSIGComm,chord} has
been cited more than ten thousand times. To highlight the wide
variety of areas in which similar allocation problems might arise, we
just mention a few more important references: content-addressable
networks \cite{CAN}, peer-to-peer systems and their associated
multicast applications \cite{pastry,scribe}.  Our algorithm is very
similar to consistent hashing, and should work for most of the same
applications, bounding the loads whenever this is desired.
In fact, our work has already found two quite different 
industrial applications; namely Google's cloud system \cite{MTZ16}
and Vimeo's video streaming \cite{Rod16}. Both systems had to handle
the lightly loaded case. Also, in both cases, load balancing
was not an objective to maximize, but rather a hard constraint, e.g.,
in the Vimeo blog post \cite{Rod16}, Rodland describes how no server
is allowed to be overloaded, and how he found a load balancing
parameter $c=1.25$ to be satisfactory for Vimeo's video steaming. We
shall return to this later.

\subsection{Background: Consistent hashing}
The standard solution to our fully-dynamic allocation problem is 
consistent hashing \cite{chord,chord-theory}. 
\paragraph{Simple consistent hashing.}
In the simplest version of consistent hashing, we hash the active
balls and bins onto a unit circle, that is, we hash to the unit
interval, using the hash values to
create a circular order of balls and bins. Assuming no collisions, a ball is
placed in the bin succeeding it in the clockwise order around the
circle. One of the nice features of consistent hashing is that it is
history-independent, that is, we only need to know the IDs of the
balls and the bins and the hash functions, to compute the distribution
of balls in bins. If a bin is closed, we just move its
balls to the succeeding bin. Similarly, when we
open a new bin, we only have to consider the balls from the succeeding
bin to see which ones belong in the new bin. 

With $\balls$ balls, $\bins$ bins, and a fully random hash function
$h$, each bin is expected to have $\balls/\bins$ balls. This is also
the number of balls we expect to move when a bin is opened or closed.

One problem with simple consistent hashing as described above is that
the maximum load is likely to be $\Theta(\log \bins)$ times bigger
than the average. This has to do with a big variation in the coverage
of the bins. We say that bin $b$ {\em covers\/} the interval of the
cycle from the preceding bin $b'$ to $b$ because all balls hashing to this
interval land in $b$. When $\bins$ bins are placed randomly on the unit
cycle, on the average, each bin covers an interval of size $1/\bins$, but we expect some bins to cover intervals of size $\Theta({\log \bins
  \over \bins})$, and such bins are expected to get $\Theta({\balls \log \bins \over
  \bins})$ balls. The maximum load is thus
expected to be a factor $\Theta(\log \bins)$ above the average.

A related issue is that the expected number of balls landing in the
same bin as any given ball is almost twice the average.  More
precisely, consider a particular ball $x$. Its expected distance to
the neighboring bin on either side is exactly $1/(\bins+1)$, so the
expected size of the interval between these two neighbors is
$2/(\bins+1)$. All balls landing in this interval will end in the same
bin as $x$; namely the bin $b$ succeeding $x$. Therefore we expect
$2(\balls-1)/(\bins+1)\approx 2\balls/\bins$ other balls to land with
$x$ in $b$. Thus each ball is expected to land in a bin with load
almost twice the average. If the load determines how efficiently a
server can serve a client, the expected performance is then only half
what it should be.

In \cite{chord-theory} they addressed the above issue using virtual bins as
described below.

\paragraph{Uniform bin covers.}
To get a more uniform bin cover, \cite{chord-theory} suggests the use
of {\em virtual bins}. The virtual bin trick is that the ball contents of
$d=O(\log \bins)$ bins is united in a single super-bin. The $d$ bins
making up a super bin are called virtual bins. 
We have only
$\bins'=\bins/d$ super bins and these super bins represent the
servers. A super bin covers the intervals covered by its $d$ virtual bins. 
The point is that for any constant $\eps>0$, if we pick a large
enough $d=O(\log \bins)$, then with high probability, each super bin
covers a fraction $(1\pm\eps)/\bins'$ of the unit cycle.

We note that many other methods have been proposed to maintain such a
uniform bin cover as bins are added and removed (see, e.g.,
\cite{BSS00,GH05,Man04,KM05,KR06}). Different
implementations have different advantages depending on the
computational model, but for the purpose of our discussion below, it
does not matter much which of these methods is used.

We also note that what corresponds to a perfectly uniform bin cover can be obtained using Rendezvous
or Highest Random Weight Hashing \cite{TR98:rendezvous}. A hash
function assigns a random weight $h(q,b)$ to each ball $q$ and bin
$b$. A ball $q$ is placed in the current bin $b$ that
maximizes $h(q,b)$. If $h$ is min-wise independent
\cite{broder98minwise}, then $b$ is uniformly random among the current
bins, corresponding to a perfectly uniform bin cover. An issue with this
scheme is that to place the ball $q$, we need to know all the current bins.

With a uniform bin cover, balls distribute uniformly between
bins.  However, with $\bins$ balls and $\bins$ bins, we still expect
many bins with $\Theta((\log \bins)/(\log\log \bins))$ balls even
though the average is $1$.  On the positive side, in the heavily
loaded case when $\balls/\bins$ is large, e.g.,
$\balls/\bins=\omega(\log \bins)$, all loads are $(1\pm
\eps)\balls/\bins$ w.h.p. However, in this paper, we want a good
load balancing for all possible load levels.

\subsection{Our algorithmic solution: Respecting bin capacities via forwarding.}\label{sec:our-forward}
We describe our full algorithm in this subsection and provide some
intuition behind its logic. 
An important part of our solution to the load balancing problem 
is to introduce a {\em capacity\/} for each bin 
that may not be exceeded by the load of the bin. In general, such a capacity
could be a
direct result of a resource constraint on the bin, but here we will use
it to enforce a strict load balancing, guaranteeing for some
balancing parameter $c>1$, that no bin has more than
$\ceil{c\balls/\bins}$ balls. 

In the description below, we will first describe the system for the
case where each bin comes with its own fixed capacity. At the end of the
section, we will elaborate on the exact vector of capacities to
achieve desirable load balancing.

Our starting point is simple consistent hashing (without virtual
bins). Bins and balls are hashed to a unit circle with a fixed hash
function that is not changed once the system is is started. For
efficiency, the hash function should be fixed at random, but the
system described below is well-defined for any given hash
function. For now, we assume that balls and bins all hash to distinct
locations so that we get a complete cyclic ordering of the balls and
bins along the cycle. A {\em ball hashes to the bin\/} following it
in clockwise order, but if the number of balls hashing to a bin exceed
its capacity, then some of the balls have to be {\em placed in other bins}.

\paragraph{Simple insertions via forwarding.}
Suppose for now that we are only adding balls to a fixed set of bins.
To deal with bin capacities, if a ball hashes to a {\em full
bin}, then it is forwarded around the circle until it finds a bin that
is not full, and then it is placed in that bin.

This forwarding is the basic idea behind {\em linear probing}. It 
appears not to have been suggested before in the context of Consistent Hashing.

\paragraph{General system invariant between updates.}
The location or placement of a ball may depend on the order in which
the balls are added, but regardless of the insertion order, the
forwarding from full bins maintains a simple invariant that we discuss
more generally before discussing other updates to the system.

If a ball $q$ hash to a bin $b$ but is located in some other bin $b'$,
then we say that $q$ {\em passes\/} all the bins from $b$ to the bin preceeding
$b'$ in clockwise order. The above forwarding from full
bins maintains.
\begin{invariant}\label{inv:forward}
No ball passes a non-full bin.
\end{invariant}
In Section \ref{sec:unique-loads} we prove that Invariant
\ref{sec:unique-loads} uniquely determines the load of each bin. This
is for a given set of balls and bins with given capacities, but
independent of history.

As balls and bins are added and removed we will maintain Invariant
\ref{inv:forward}. More precisely, after each of these updates, we
assume we have time to {\em settle the system}, moving balls so as to
restore Invariant \ref{inv:forward} and satisfy all capacity
constraints before we get another update to the system.

\paragraph{Searching a ball in the bins}
Assuming Invariant \ref{inv:forward}, it is
easy to find a ball $q$ among the bins. First we search the bin
$b$ that $q$ hashes to.  If $q$ is not found and if $b$ is full, we move to
the next bin $b'$ and recurse until either we find $q$ or we reach a
bin $b'$ without $q$ that is not full.  In the latter case, by
Invariant \ref{inv:forward}, we conclude that $q$ is not in the
system. If $q$ is to be inserted, it should be inserted in $b'$.

\paragraph{General insertions and forwarding from overfull bins} 
Above, when we inserted a ball, we just placed it in the first
non-empty bin. We refer to this as {\em simple insertions}. We shall
here allow for more flexible non-deterministic insertions. 

Temporarily, we allow a bin to be {\em overfull\/} in the sense that
the load exceeds the capacity.  When a ball arrives, it is first
placed in the bin it hashes to, which may become overfull. If the
system has an overfull bin, we forward {\em any\/} of its balls to the
next bin.  

In relation to Invariant \ref{inv:forward}, we view an
overfull bin as full, not non-full. This means that Invariant
\ref{inv:forward} is never violated in the above process, so we
are done when there are no overfull bins.

The previously described simple insertions correspond to
always forwarding the most recent arrival, but now we are free to
choose which ball to forward. For example this allows a {\em history
  independent\/} system using the idea from \cite{BG07}. It assumes a
linear order on all balls which is independent of the order in which
they are inserted.  When placing balls in bins, we place them as if
they were inserted lowest to highest using simple insertions. To
maintain this history independence with general insertions, we always
forward to the highest balls from an overfull bin.

When we later analyze the efficiency of our system, we will allow for
general insertions with no restrictions on which balls get forwarded
from overfull bins. 

We will generally use the number forwardings from overfull bins as an
upper bound on the number of balls moved, but note that the actual number
of moves may be much smaller, e.g., with the simple insertion, we are just
moving the inserted ball directly to the first non-full bin.

\paragraph{Removing a bin}
Removing a bin has the same effect as setting its capacity to 0. It is
now overfull, and then we forward balls from overfull bins arbitrarily until
no bin is overfull. Invariant \ref{inv:forward} was never violated.

\paragraph{Deleting balls and filling holes}
When we delete a ball from a bin, we may have to replace it. It is
convinient to talk about it as filling holes, where the number of {\em
  holes\/} in a bin is the number of balls it is missing to get
full. When we delete a ball, we create a new hole. 

Our concern about a new hole in a bin $b$ is if $b$ was full and
passed by a ball $q$ in some later bin $b'$. The hole renders $b$
non-full and then Invariant \ref{inv:forward} is violated. We can then
fill the hole by moving $q$ to $b$, but this creates a new hole in $b'$
that we may have to fill recursively. Invariant
\ref{inv:forward} is restored when no hole is passed. The process
never violates any capacity constraint, so we are done.

Inspired by deletions in linear probing, an efficient way to fill holes, starting from the first hole in some
bin $b$, is to scan the bins one by one. When we get to a bin $b'$
with a ball $q$ that hash to or before $b$, we fill $b$ with $q$, and
proceed from $b'$. This way we never consider the same bin twice, and
we can stop when we meet a non-full bin. We shall discuss an even more
efficient implementation in Section \ref{subsec:computemoves} where we
maintain the number of balls passing each bin. This allows us to stop
as soon as we create a hole in a bin that is not passed. For now, however,
we are focussed on bounding the number of moves.

\paragraph{Adding a bin} When a bin is added, it starts with as as many holes
as its capacity. We then keep filling holes as long as possible so  
Invariant \ref{inv:forward} is restored. No capacity constraint is violated
in this process.

\paragraph{Changing Capacities for Load balancing in Dynamic Environments.}
We will now show how to set and change capacities to achieve good load
balancing. For a given load balancing parameter $c=1+\eps>1$, we want
to guarantee that no bin has more than $\ceil{c\balls/\bins}$ balls.
One possibility would be to just say that all bins had capacity
$\ceil{c\balls/\bins}$, but then adding a single ball could force us
to increase the capacity of all bins, completely changing the
configuration. As a result, we need to be careful about enforcing the
above capacity constraints across all bins.  In particular, to
minimize the number of capacity changes when balls are inserted or
deleted, we aim for a total bin capacity of $\ceil{c\balls}$.
Assuming a linear ordering of the bins, we let the lowest 
$\ceil{c\balls}-\bins\floor{c\balls/\bins}$ bins have
capacity $\ceil{c\balls/\bins}$ while the rest have capacity
$\floor{c\balls/\bins}$.  We refer to the former bins as {\em big
  bins} and the latter bins as {\em small bins}, though the difference
is only 1.  Moreover, as an exception to the above rule, we will never
let the capacity drop below $1$, that is, if $c\balls<\bins$, then all
bins have capacity $1$.  

For a given set of balls and bins, with a given linear order on
the bins, the above protocol tells us the exact capacity of each
bin. When we add or removing a ball, it affects the capacity of at 
most $\ceil{c}$ bins: if the ball is added, we increment 
capacity, and if it is deleted, we decrement capacity. If a bin is 
inserted or deleted, it changes at most at 
$\ceil{c\balls/\bins}$ capacities. If capacities are increased, we may
have to fill holes, and if capacities are decreased, we may have to
forward balls.

\paragraph{Separating capacity changes from ball and bin updates.}
For the sake of later analysis and efficiency, we make the rule that
if a ball is inserted or a bin is deleted, we first do all the
required capacity increases, one by one, settling the system after
each capacity increase by hole filling. When the capacities have
all settled, we add the ball or remove the bin, upon which we do the final
settlement by forwarding from overfull bins.

Conversely, when a ball is deleted or a bin is inserted, we settle
by hole filling before we decrease any capacity. Next we do
the required capacity decreases one by one, forwarding from
overfull bins after each, until we are back to the desired 
settled system satisfying Invariant \ref{inv:forward} and all capacity
constraints.

\paragraph{Concrete random hash functions} 
The system as described above always works correctly with a fixed hash
function, but an adversary knowing our hash function could force us to
make way too many moves. In order to claim efficiency, we use a hash
function that is fixed randomly and independent of the operations
performed on the system. Thus we can think of the updates and searches
for balls in bins as chosen in advance by an adversary that is
oblivious to what goes on inside the system.

In practice, we work with hash functions with a limited range
$[r]=\{0,....,r-1\}$. Mapping this range to a circle, position $0$
succeeds $r-1$.  We assume that the hash function for balls is independent of the hash function 
for bins. Unless otherwise stated, the bounds in this paper
only assume that $r\geq \bins$ and that both
the ball hash function and the bin hash function are
5-independent (simple tabulation \cite{patrascu12charhash} will
also work even though it is only 3-independent). 
With limited-range hash functions, we may have collisions. To get a
complete cyclic order, we do the following tie breaking: if two balls
or two bins hash to the same location, then the one with the lower ID
precedes the one with the higher ID. Moreover, if a ball and a bin
hash to the same location, the ball precedes the bin. This implies that
the bins hashing to a given position $x$ will always be filled bottom-up.

\subsection{Our Results: Theoretical and Empirical}
In this subsection, we state our main theorems and other
results on the above system.  Subject to the capacity constraints, our
main focus in this paper is the expected number of balls that have to
be moved when a ball or bin is inserted or deleted. Our bounds will
hold no matter which balls we decide to move when a bin is overfull
or a hole has to be filled, as long as we follow the above
description. We shall discuss how to efficiently identify the balls to
be moved in Section~\ref{subsec:computemoves}.

Mathematically, the most interesting case is when $c=\omega(1)$. We
note that inserting a ball results in up to $\ceil{c}$ bins increasing
their capacity. Nevertheless, besides placing the new ball, we will
prove that the expected number of ball moves is $O((\log c)/c)=o(1)$.

The general result is 
\begin{theorem}\label{thm:moves-c}
For a given load balancing 
parameter $c\geq 2$,  
the expected number of
bins visited in a search is $1+O((\log c)/c)$. When a
ball is inserted or deleted, the expected number of other balls that
have to be moved between bins is $O((\log c)/c)$. When a bin is
inserted or deleted, besides moving $O(\balls/\bins)$ expected balls hashing
directly to the bin, we expect to move $O((\balls/\bins)(\log c)/c)$ other
balls.
\end{theorem}
For the insertion and deletion of bins, Theorem \ref{thm:moves-c}
implies that the expected number of moves is $O(\balls/\bins)$. We distinguish the balls hashing directly to the bin so as
to allow a direct comparison with 
simple consistent hashing without capacity constraints
\cite{chord-theory}. For simple consistent hashing,
the balls affected by the insertion or deletion of a bin are exactly
the balls hashing to it. We expect $O(\balls/\bins)$ such balls (previously,
we have said that exactly $\balls/\bins$ 
balls were expected to 
hash to any given bin, but
that was assuming ideal fully random hash functions). With our
capacity constraints, we only expect to move $O((\balls/\bins)(\log c)/c)$
other balls. The price we pay for guaranteeing a maximum
load of $\ceil{c\balls/\bins}$ is thus only a multiplicative factor $1+O((\log c)/c)=1+o(1)$ in the expected number of 
ball moves. 

%

For $c\in(1,2]$, we parameterize by $\eps=c-1>0$.
\begin{theorem}\label{thm:moves-eps}
For a given load balancing parameter $c=1+\eps\in(1,2]$, the expected
  number of bins visited in a search is $O(1/\eps^2)$. When a ball is
  inserted or deleted, the expected number of other balls that have to
  be moved between bins is $O(1/\eps^2)$.  When a bin is inserted or
  deleted the expected number of balls that have to be moved is
  $O(\balls/(\bins\eps^2))$.
\end{theorem}
The bounds of Theorem \ref{thm:moves-eps} are similar to those
obtained in \cite{patrascu12charhash} for linear probing. The
challenge here is to deal with the fact that bins are randomly placed,
as opposed to linear probing where every hash location has a bin of
size 1. Nevertheless we will be able to reuse some of the key lemmas
from \cite{patrascu12charhash}. The proof of Theorem \ref{thm:moves-c}
is far more challenging, and the main focus of this paper.

\begin{remark} The bounds from Theorems \ref{thm:moves-c} and
\ref{thm:moves-eps} also hold in the simpler case where all bins have
a fixed capacity $C$ and we define $c=1+\eps=C\bins/\balls$.  We
note that our updates change the value of $\balls$ and $\bins$, hence
of $c=C\bins/\balls$. For the bounds to apply, we always use the
smaller value of $c$ in connection with each update. Thus, for the
bounds on the moves in connection with a ball insertion or bin
removal, we use the value of $c$ before the update.  For the bounds on
the moves in connection with a ball deletion or bin addition, we use
the value of $c$ after the update. 
\end{remark}

\paragraph{Additional results.} In this paper, we are also going
to discuss how to efficiently find the balls to be moved in connection
with system updates. We are also going to discuss high probability
bounds that are, essentially, a factor $O(\log n)$ worse than the
above expected bounds. Finally, we note that in practice, it may be
relevant to study our problem with weighted balls. 
In this case, the sum of the weights of the balls assigned to a bin should not exceed its capacity. 
With weights, the results get less
clean, e.g., a single weight might be bigger than the allowed
capacity. Also, there may be situations where we have to move a large
number of light balls to maintain balance. However, if we use integer
capacities as described above, and if no balls
has weight above $1$, then the above results hold if we replace ``number
of balls'' with ``weight of balls''. We note that in the non-weighted case, the probability bounds in this paper are about sums of binary $0-1$ variables. However, the extend easily to sums of real variables in range $[0,1]$. To avoid corner cases, we should assume no ball has weight zero, i.e. all weights are bounded below from zero by some constant.

\paragraph{Empirical Results and Industrial use.}
We confirm effectiveness of our algorithm in various practical
settings via an empirical study
in section~\ref{sec:simulation}. Furthermore, we note that
versions of our algorithm
have been deployed in a number of industrial applications.
The first industrial application of our algorithm was in Google's
Cloud Pub/Sub as described in \cite{MZ17:google}. This application
needed good load balancing in a variety of instances with different 
characteristics including the lightly loaded case. History
independence was also necessary in that application. 

Surprisingly, within a few months of the first release of this paper
on arXiv \cite{MTZ16}, our algorithm got picked up by the video
streaming company Vimeo. In a Vimeo blog post \cite{Rod16}, citing
our paper, Rodland
describes how Vimeo used our algorithmic idea to solve the scaling issues
they had in handling almost a billion requests per day (they have 170
million users). They implemented a version of algorithm with balancing
parameter 1.25, and then the problems vanished as seen in Figure
\ref{fig:vimeo}. In \cite{Rod16}, they also explain that they had
given up on traditional consistent hashing and power of two choices
because of load balancing problems. Finally \cite{Rod16} explains that
an open source version of the code has been realased HAProxy 1.7.0,
and we expect that other companies will start using it.

\begin{figure}
\centerline{\includegraphics[width=10cm]{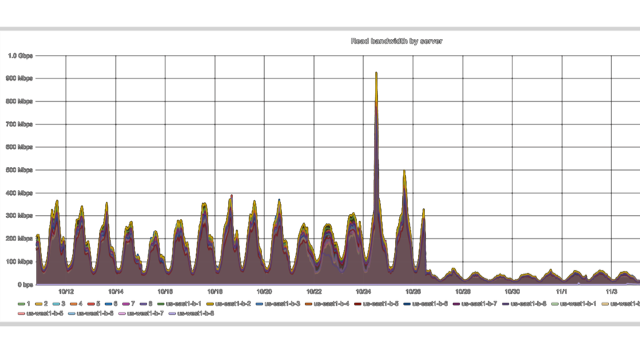}}
\caption{This figure is copied from \cite{Rod16}. It shows Vimeo's use
  of shared cache bandwidth over time. The bandwidth problems
  disappeared when they switched to a version of our algorithm on October
  26. \label{fig:vimeo}}
\end{figure}

The Vimeo story is a good example of how theory can have impact. We
offer a simple algorithmic solution to the load balancing
issue. We are not claiming to have a theoretical model that captures
all the important aspects of performance since it depends on the concrete
implementation context. What we do offer is a theoretical analysis,
showing that for every possible input, the algorithm has very good
expected performance on important combinatorial parameters, including
guaranteed load balancing for all possible inputs.  This is
something that can never be verified by tests, but it is very valuable
for the trust in an online dynamic system that is meant to run
``forever'' not knowing future inputs.

\paragraph{Note about simplicity and general applicability.} 
Consistent hashing is a simple versatile scheme that has been
implemented in many different systems with different constraints
and performance measures~\cite{OV11,GF12}.  Our consistent hashing scheme, respecting
capacities via forwarding, is almost as simple, and should work with
most of these implementations.
The classic implementation of consistent hashing is the distributed
system Chord \cite{chordSIGComm,chord} which has more than ten thousand
citations. The Chord papers \cite{chordSIGComm,chord} give a thorough description of
the many issues affecting the design. Below we only consider
the aspects affecting the message complexity of finding a ball among the
bins.

In Chord, they have a system of pointers so that given an arbitrary
point on the cycle, they can find the next bin in the clockwise
order using $O(\log n)$ messages. This is how they find the bin a ball
hashes to. In simple consistent hashing, this is where the ball is to
be found.  With our forwarding, starting from the bin a ball hashes to,
we need to visit succeeding bins until we either find the ball or a
non-full bin. Chord does maintain explicit successor pointers between
neighboring bins, so we only have to pay $O(1)$ extra messages per bin
visited. By Theorem \ref{thm:moves-eps}, we expect to visit
$O(1/\eps^2)$ bins, so our total expected message cost is $O(\log
n + 1/\eps^2)$. The extra $O(1/\eps^2)$ is negligible if
$\eps=\omega(1/\sqrt{\log n})$.

Of course, there may be other systems with different costs associated,
but generally we assume it to be significantly more expensive to find
the bin a ball hashes to than it is to go from one bin
to its succeeding neighbor. 
Since the costs depend on the implementation context,
our focus here is on fundamental combinatorial properties like
loads, the number of bins searched, and the number of balls that are
moved when the system is updated.

\drop{
\paragraph{Computing moves when the system is updated}
So far, we have not been concerned with the time it takes to compute
which clients have moved when the system is updated. 
We can, in fact, identify these moves in expected time proportional
to the above bounds on expected
number of moves. 
This is most challenging to prove in the case where $c=\omega(1)$ and
$c\balls/\bins=O(1)$. In this case, when a ball is inserted, we
have only $O(1)$ time, and this has to include $\Theta(c)$ bin capacity 
increases. To get down to constant time, we relax the constraint on how many bins that
have capacity $\ceil{c\balls/\bins}$, and how many that have capacity $\floor{c\balls/\bins}$.
We elaborate more in Section~\ref{subsec:computemoves}.}

\subsection{Power of multiple choices as an alternative?}\label{sec:power}
We have proposed applying the idea of forwarding in the style of linear probing to
bound the loads in consistent hashing. The reader may be wondering if alternative ideas 
in dynamic load balancing can be applied to deal with our natural problem. 
Here, we discuss them, and highlight their drawbacks or challenges to convert them to
working solution.
In particular, we discuss the possibility of instead using the power of multiple choices
\cite{power-of-two-choices,cuckoo,mitzenmacher2001power,berenbrink,talwar14,Voc03}. The discussion is speculative in that none of these
techniques have been analyzed in our dynamic context where both balls and
bins  can be added and removed. We can try to guess what may happen,
but have limited knowledge without a proper analysis.

The most basic form is that we have a fixed set of
$\bins$ bins. Balls are added, one by one. Each ball is given $d$
uniformly random bins to choose from, and picks the least loaded, breaking
ties arbitrarily. With
$\balls$ balls, w.h.p., we end up with a maximum load of
$\balls/\bins+\frac{\ln\ln \bins}{\ln d}+\Theta(1)$ \cite{berenbrink}.
An interesting twist suggested by V{\"{o}}cking is that if several
bins have the same smallest load, we always pick the left-most
\cite{Voc03}. Surprisingly, this reduces the max load to
$\balls/\bins+\Theta(1+\frac{\ln\ln \bins}{d})$ \cite{Voc03}.
We note that to get a constant ratio between maximum and average load
when the average load is constant, we do need a super constant number
of choices, e.g., $d=\Omega(\log\log \bins)$ with left-most choice.

The above mentioned bounds are proved in the ideal case where
we pick uniformly between the bins.  Consider now first the case of
simple consistent hashing where both balls and bins are placed at
random on a unit circle, and a ball goes to the succeeding bin in
clockwise order. This case was studied in \cite{BCM03,BCM04}, where it
was proved that if $\balls=\bins$, then the maximum load is
$O(\log\log \bins)$. However, with a concrete example, \cite{Wie07}
showed that we cannot in general hope to get max-load
$\balls/\bins+O(\log\log \bins)$ when $\balls\gg\bins\log\log\bins$.
This is again for the case of simple consistent hashing where bins are
just placed randomly on the circle. However, using, e.g., virtual
bins, we know that that we can obtain a more uniform bin cover such
that each bin represents a fraction $(1\pm\eps)/\bins$ of the unit
cycle and where a ball lands in a bin with this probability. With
$\eps=1/2$, the main result from \cite{Wie07} implies that using the
power of $d$ choices, w.h.p., we get a maximum load of
$\balls/\bins+O(\frac{\log\log \bins}{\log d})$.

We still have to consider what happens in our dynamic case where balls
may be deleted and bins can be added and/or removed. The results from
\cite{CFMMRSU98} indicate that to delete a ball, it may suffice to
just remove it without moving any other balls. However, if a bin 
is removed, we have to move all its balls. In order to claim any
bounds, we would need a careful analysis, but the best we can possibly
hope for is to match the above bounds for the basic case where $\balls$ 
balls are just added picking uniformly between $\bins$ fixed bins.

\paragraph{Cuckoo hashing}
We will now discuss how the power of multiple choices could possibly be used
in the style of Cuckoo hashing \cite{pagh04cuckoo} to provide 
balancing guarantees like those in our solution.
With user specified load balancing parameter $c=(1+\eps)$, we
assigned capacities $\floor{c\balls/\bins}$ or $\ceil{c\balls/\bins}$
to each bin in such a way that the total capacity was
$\ceil{c\balls}$ and such that we only changed at most $\ceil{c}$
bin capacities in connection with each ball update.
We then used forwarding in the style of linear
probing to respect these capacities. Obviously, one could try to adapt
many other hash table schemes to respect these capacities. Most
notably, we could hope to adapt Cuckoo hashing \cite{pagh04cuckoo}.

We now review the basic results for Cuckoo hashing in the ideal case
where we have a fixed set of $\bins$ bins, all of the same capacity
$C$, and where each of $\balls$ balls gets $d$ uniformly independent
choices between these bins. The basic feasibility question is
how large $C$ has to be before we expect to be able to place all
balls. This question is studied in \cite{FKP16}, but here
we also want to update the system efficiently as balls and bins are
inserted and remvoed.  In the original Cuckoo hashing
\cite{pagh04cuckoo}, the capacities are all 1. With 2 choices, the
balls are placed into $\bins=(2+\eps)\balls$ bins for any
constant $\eps$. However, \cite{FPSS05} proves that using
$d=O(\log(1/\eps))$ choices, we can place the balls in only
$\bins=(1+\eps)\balls$ bins, so the maximum load of $1$ is at most
$(1+\eps)$ times the average load, as desired for load balance
$(1+\eps)$.  For efficient insertions, \cite{FPSS05} prove that if
$d\geq 5+3\ln(1/\eps)$, then a ball can be placed in $d^{O(\log
  (1/\eps))}=(1/\eps)^{O(\log\log(1/\eps))}$ expected time
and in $o(n)$ time with high probability. 
The high probability bound on the insertion time is improved to a polylog in \cite{FPS13}. Moreover, using a larger value of $d=
O\left(\frac{\log (1/\eps)}{\eps}\right)$, it was recently
shown \cite{FJ17} that the expected search time can be reduced to $O(1)$
independent of $\eps$. 
  In \cite{DW07} they consider
having only $d=2$ choices, but instead they increase the
capacity. They show that a bin capacity of size $O(\log(1/\eps))$
suffices to distribute the balls with load balance $(1+\eps)$. For
efficient insertions, they prove that if the capacity is above
$16\ln(1/\eps)$, then a ball can be placed in
$\log(1/\eps)^{O(\log\log (1/\eps)}$ expected time. The bounds
require fairly complex hash functions, but \cite{DW07} states
that the hash functions can be implemented in $O(\bins^{5/6})$ space. 
A further discussion of hash functions for Cuckoo hashing is found in
\cite{ADW14}.

We note that it is would require a full analysis to figure out
to which degree the above bounds transfer to our fully dynamic case
where bins can be added and removed, and where capacities adapt to the
load based on the balancing parameter. Also, we would most likely need
to use virtual bins to get a reasonably uniform bin cover. The best we
can hope for is to match the above bounds for the simpler
case with uniformly distributed choices between a fixed set of bins.
This may indeed be possible, and would be interesting.

\paragraph{Stepping back,} the first contribution of this paper is to raise
the problem of getting consistent hashing to obey a user specified load 
balance paremeter $c=(1+\eps)$ for all possible inputs.
We present the first solution with proven bounds 
on the efficiency of searches, ball and bin updates. When $\eps$
is a positive constant, our expected bounds are only a constant factor from
the general lower bounds. 

It may very well be possible to get a different solution based on
Cuckoo hashing instead of linear probing, getting bounds similar to
those reviewed above for a fixed set of bins, all with the same
capacity.  However, even if this can be done, it would not always be
better than our solution. One of the attractions of $d\geq 2$ choices
is that a ball only has to be searched in $d$ bins, but these are $d$
random bins, as opposed to the single segment of consecutive bins
searched with our solutions. If we, as Vimeo, use balancing parameter
$1+\eps=1.25$, then, by Theorem \ref{thm:moves-eps}, we only
expect to consider a constant number of consecutive bins. Thus, if, as
in the Chord implementation, the cost of accessing a random bin is
much bigger than that of getting from a bin to its successor, then our
search is nearly twice as efficient as checking 2 random choices (for
successful searches, in expectation, we can bring 2 down to 1.5). With
Cuckoo hash tables, it is often seen as an advantage that the $d$
choices can be checked in parallel, but here we imagine a large system
processing many requests in parallel, and then the total number of
messages is important. With updates we have the same advantage
of only working within a single segment of bins whereas Cuckoo hashing has to consider
arbitrary bins. However, the Cuckoo updates are still to be defined
and analyzed for us to understand their efficiency when bins are inserted and deleted and
capacities change. Moreover, for the implementation of Cuckoo hashing, 
we probably need the added complexity of, e.g., virtual bins to get a more
uniform bin cover. Thus, we expect our linear forwarding solution to remain relevant
regardless of future developments with multiple choices. The situation
is like linear probing versus Cuckoo hashing for regular hash
tables. Both are very important solutions to an extremely important
problem, and each has its advantages in different contexts, e.g.,
with Cuckoo hashing, the queries are easy while the updates are more 
challenging. For our problem of 
consistent hashing with bounded loads, we prove here
that the forwarding of linear probing works in a simple
practical solution. A corresponding
analysis for Cuckoo hashing is yet to be done.

\section{High Level Analysis}\label{sec:highlevel}
To analyze the expected number of moves, we shall use a general
probabilistic understanding of configurations encountered. By a
configuration, we refer to the situation between updates where
the allocation has settled to satisfy Invariant \ref{inv:forward} and with
no overfull bins.

\subsection{Uniqueness of loads}\label{sec:unique-loads}
As a first step in understanding configurations, we argue that
Invariant \ref{inv:forward} uniquely determines the load of all
bins. This is for a given set of balls and bins with given capacities,
but independent of history. This is a very simple generalization of
the argument that the set of cells that are filled by linear probing is 
unique and history independent.

First we note that Invariant \ref{inv:forward} implies that which bins are
full is independent of the order in which the balls are
inserted. It only depends on which balls and bins are present and where
they hash to. More precisely, we have
\begin{lemma}\label{lem:fulls}
A bin $b$ is full if and only if there is an interval of consecutive bins
$B=b_1,...,b_k$ ending in $b=b_k$ such that the total number of balls
hashing to these bins is at least as big as their total capacity.
\end{lemma}
\begin{proof} If $b$ is full, we take $B=b_1,...,b_k$ to be the maximum
interval of full bins, that is, the bin $b'$ preceding $b_1$ is not
full.  By Invariant \ref{inv:forward}, this means that no balls
hashing to or before $b'$ can end in $B$, so $B$ must be filled with
balls hashing to $B$. In the other direction, the
result is trivial if all balls hashing to $B$ end in 
$B$, since there are enough balls to fill all bins. However,
if a ball hashing to $B$
ends up after $b_k$, then $b_k$ is full by Invariant \ref{inv:forward}.
\end{proof}
In fact, the hashing of balls and bins determines completely the number
of balls landing in any given bin $b$. If the bin is full, this follows
from Lemma \ref{lem:fulls}. Otherwise, the number of 
balls landing in $b$ is determined by Lemma \ref{lem:non-fulls} below.
\begin{lemma}\label{lem:non-fulls}
If a bin $b$ is not full, consider the longest interval $b_1,\ldots,b_k$
of full bins leading to $b$, that is, $b$ succeeds $b_k$   and the
predecessor $b'$ of $b_1$ is not full. If the predecessor of $b$ is
empty, with have $k=0$, hence no bins in $b_1,\ldots,b_k$.

Then the balls landing in
$b_1,\cdots,b_k,b$ are exactly the balls hashing to these bins.
If $s$ balls hash to $b_1,\cdots,b_k,b$ and $b_1,\cdots,b_k$ have capacities
$C_1,\cdots,C_k$, then we have exactly $s-\sum_{i=1}^k C_k$ balls landing
in $b$.
\end{lemma}
\begin{proof} The result follows by Invariant \ref{inv:forward} together
with the fact that $b$ and $b'$ are not full.
\end{proof}
Summing up the two previous lemmas, we have
\begin{lemma}\label{lem:bin-loads}
The set of balls and the set of bins with their capacitities uniquely
determine how many balls land in each bin.
\end{lemma}

\subsection{Expected distance to non-full bin}
Suppose we have $\balls$ balls and $\bins$ bins that are currently in
the system.  We refer to these balls and bins as {\em active}. We will
also study {\em passive\/} balls and bins that are not currently
in the system, yet which have hash values that will be used if they
get inserted. For some $\bar c>1$, the total capacity will be
exactly $\bar c\balls$.  Since no bin has capacity below $1$, we
always have $\bar c\balls/\bins\geq 1$. In our analysis, we will only
assume that each bin has a capacity between $\bar c\balls/(2\bins)$
and $2\bar c\balls/\bins$, and that the concrete capacites are independent
of the hashing of balls and bins.  We note that this is always satisfied when
bin capacities are at least $1$ and differ by at most $1$.

Theorem \ref{thm:exp-bins-c} below gives our main technical understanding of
configurations for larger $c$. It does not make any assumptions about
how we reached the configuration as long as Invariant \ref{inv:forward}
is satisfied with the desired bin capacities. Theorem  \ref{thm:exp-bins-c} 
may seem too complicated, but this is exactly the technical challenge of this paper: {\em the devil in
the details}. We need to exploit every detail of the theorem
to get the bounds we have claimed in the introduction. The proof
of Theorem   \ref{thm:exp-bins-c} is based on a very delicate analysis
of the interaction between the balls and bins.
\begin{theorem}\label{thm:exp-bins-c} Consider a configuration with $\balls$ active balls and $\bins$ active bins 
and total capacity $\bar c\balls$ for some $\bar c\geq 2$. Suppose, moreover, that
each bin has capacity between $\bar c\balls/(2\bins)$ and $2\bar c\balls/\bins$. Then
\begin{itemize}
\item[(a)] 
Starting from the hash location of a given passive ball or active or passive bin,
the expected number of consecutive full bins is $O(1/\bar c)$. 
\item[(b)] If we start
from a given active bin of capacity at least $2$, the expected number of 
consecutive full bins is $O((\log \bar c)/\bar c^2)$.
\item[(c)] The expected number of balls hashing directly to any given active bin
is $O(\balls/\bins)$. The expected number of balls forwarded into the bin
is $O((\balls/\bins)((\log\bar c)/\bar c^2))$. Finally, if a bin $i$ is not active, and its active
successor $i'$ is given an extra capacity of one, then the expected number
of full bins starting from $i'$ is $O((\log \bar c)/\bar c^2)$.
\end{itemize}
The above statements are satisfied if the balls and bins are hashed
independently, each using 5-independent hash functions or simple
tabulation hashing. The statement of (c) may seem a bit cryptic, and will make more sense in the context of the analysis it is used in below.
\end{theorem}
The worst-case for our bounds is when the capacities are $1$ and
$2$. This case explains why (a) would not work for an active ball
since an active ball by itself could fill a bin of capacity
1. However, when a ball $q$ is inserted, we do forwarding to the
nearest non-full bin in the configuration before $q$ is insertion where
$q$ is still passive, and therefore Theorem \ref{thm:exp-bins-c} (a) gives
the expected number of fulls bins passed.

\begin{corollary}\label{thm:find}
With balancing parameter $c=1+\eps\in(1,2]$, the expected number of
bins visited in a search is $1+O((\log c)/c)$ if
$c\geq 2$, and $O(1/\eps^2)$ if $c\leq 2$.
\end{corollary}
\begin{proof} 
If the ball $q$ searched is not in the system, then the
search is to only up to and including the first non-full bin. However,
if $q$ is in the system, the latest it can be placed is if was
added last, which corresponds is the first non-full bin if in the
system without $q$. So in both cases, the expected
search is bounded as one plus the number of consecutive full bins
starting from an passive ball as in Theorem \ref{thm:exp-bins-c} (a).
\end{proof}

\subsection{Bounding the expected number of moves}
We are now going to prove Theorem \ref{thm:moves-c} applying
Theorem \ref{thm:exp-bins-c} to settled configuations. If
the settled configuations has $\balls$ active balls, $\bins$
active bins, and total capacity $\bar c m$, then we
need to make sure that $\bar c\geq c$ where $c$ is the fixed parameter chosen 
to control the bin capacity.
The bounds only get better with larger $\bar c$. We know that that
before and after every update we have total capacity $\ceil{c
\balls/\bins}=\bar c\balls/\bins$ so $\bar c\geq c$. However, 
we will also argue about intermediate configurations, e.g., when
we insert a ball, we first implement capacity increases one by one,
settling into a valid configuration after each capacity increase, but
these capacity increases will only increase $\bar c$ since $m$ has
not increased yet.

The moves have already been described in Subsection \ref{sec:our-forward},
but we review them below for the sake of analysis.

First we discuss plain inserations and deletions without any changes
to bin capacities.

\paragraph{Plain insertion}
Consider an insertion of a ball $q$ hashing to some bin $i$. For now,
we ignore that some capacities have to be increased. First we place
$q$ in $i$, which may become overfull. If so, we forward some ball to
the succeeding bin $i'$, repeating until we reach a bin that has room
because it was not full. The important thing to note is that we move at most one ball per
full bin encountered, and that we stop when we meet a bin that was not
full before the insertion. 
By Theorem \ref{thm:exp-bins-c} (a) applied to the configuration
before the insertion, the expected number of moves, excluding the
initial placement, is $O(1/\bar c)=O(1/c)$.

\paragraph{Plain deletion}
When we delete a ball $q$ from a bin $i$, we create a hole in $i$ that
we try to fill with a ball $q'$ that has passed bin $i$ in the sense
that $q'$ resides in a later bin $i'$ but hashes to bin $i$ or some
earlier bin. If we succeed, we recursively try to fill the new hole
left by $q'$ in $i'$. Suppose the last hole created is in bin
$i''$. We have now completed the deletion with a new valid
configuration. We know that all bins from $i$ to the bin preceding
$i''$ are full while $i''$ is not full. The number of moves, including the
initial deletion of $b$, is bounded by
the number of bins from $i$ to $i''$. By Theorem \ref{thm:exp-bins-c}
(a) applied to the configuration after the deletion, it follows that
the expected number of moves, including the final removal, is
$1+O(1/\bar c)$ but now $\bar c$ is measured right after the deletion
where we have only $\balls-1$ balls so $\bar
c=\ceil{c\balls/\bins}m/(\balls-1)>c$.

\paragraph{Full deletion} 
To complete a deletion, we have to do $c\pm 1$ capacity decreases.
Since the ball has been removed, we have $\balls'=\balls-1$ balls
during the capacity decreases that  all
together bring the the capacity down to $\ceil{c\balls'}$. Doing
one capacity decrease at the time, we know that every configuration
encountered has total capacity $D\geq \ceil{c\balls'}$, and hence
$\bar c=D/\balls'\geq c$. We recall that the bins picked for capacity
decreases are chosen independently of the random choices made by the
hash functions.

When we decrease the capacity of a bin $i$, if the bin now has one
ball more than its capacity, we perform exactly the same process as
when we hashed a ball to a full bin, forwarding a ball until
we reach a bin that was not already full.  The number of moves is
bounded by the number of consecutive full bins starting from $i$ in
the configuration before the capacity of bin $i$ was decreased.

A crucial observation is that since the lowest possible capacity is $1$,
bin $i$ had capacity at least two before the decrease, so by
Theorem \ref{thm:exp-bins-c} (b) applied to
the configuration before the capacity decrease, the expected number of moves
is $O((\log\bar c)/\bar c^2)=O((\log c)/c^2)$. We are doing at most $\ceil c$
such capacity decreases, so the total expected number of moves, including
the plain deletion of $q$ itself, is bounded by
\[1+O(1/c)+(1+c)O((\log c)/c^2)=1+O((\log c)/c).\]

\paragraph{Full insertion} 
A full insertion is symmetric to a full deletion, doing
$c\pm 1$ capacity increases before we do the plain insertion. The
expected number of moves, including the 
the addition of $q$ itself is $1+O((\log c)/c)$.

\paragraph{Bin updates and forwarding}
When we do bin updates, there may be a lot possibilities for which
balls to move. We will argue that the concrete choices do not matter
for our overall bounds.

Note that for the placement of balls in bins according to Invariant
\ref{inv:forward}, we can think of all passive bins as active bins
with capacity zero. Then adding a bin with capacity $C$ is like
increasing the capacity to $C$, while removing it is like decreasing
its capacity to $0$. The lemma below considers the general impact of
changing the capacity of a bin.

\begin{lemma}\label{lem:mov-cap-dec}
Suppose we have a valid configuration, and decrease the capacity of
some bin $i$ from $C^+$ to $C^-$ while leaving all other capacities
unchanged. We assume that the total capacity remains strictly bigger than
the total load. Then the total number of forwardings done from overfull
bins to reach a new valid configuration depends only on which
balls are in the system.
The number of forwardings also
bounds the maximal number of moves done if we conversely increase the capacity
from $C^-$ to $C^+$, filling holes to restore Invariant \ref{inv:forward}.
\end{lemma}
\begin{proof} The proof of Lemma \ref{lem:mov-cap-dec} is fairly standard,
thinking of forwarding as a flow of balls.

We know from Lemma \ref{lem:bin-loads} that the balls and bins with
capacities uniquely determine the load of any bin. In particular,
this determines all bin loads before we do the capacity decrease.

When the capacity of bin $i$ is decreased to from $C^+$ to $C^-$, bin $i$
may become overfull, and then we have to forward some balls.
Then, repeatedly, we take an arbitrary ball from any overfull bin and
forward it to the next bin, until no overfull bin remains.  When we
forward from a ball from an overfull bin it remains at least full, so
if $i'$ is the last bin we forward from, then all bins from $i$ to
$i'$ are full. Since the total capacity remains strictly bigger than the
total load, we cannot keep forwarding the whole way around the cycle.

Let $i_0,\ldots,i_k$ be the bins from $i$ to $i'$.
For $j=0,\ldots,k-1$, let $Q_j$ be the set of balls that got forwarded
from $i_j$. Then $Q_0$ are all balls forwarded from $i_0$ so $|Q_0|$ 
is exactly the number of balls with which $b_0$ exceeded the new capacity $C^-$.
For $j=1,\ldots,k-1$, $Q_j$ consists of balls that either
started in bin $i_j$ or came from $Q_{j-1}$. Since bin $i_j$
ends up full, $|Q_j|-|Q_{j-1}|$ is exactly the number of balls bin $i_j$ 
needed to fill its capacity. We have $|Q_j|=|Q_{j-1}|$ if bin $i_j$ was
full before the capacity decrease; otherwise $|Q_j|<|Q_{j-1}|$. 
The process stops at bin $i_k$ because adding $|Q_{k-1}|$ balls to bin $i_k$
does not exceed its capacity. 

From the above description, it follows that the numbers $|Q_0|,\ldots,|Q_{k-1}|$ are
uniquely determined by the bin loads before the capacity decrease, and
these were uniquely determined by the balls and bins with capacites.
In particular, this determines
the total number of forwardings  $\sum_{j=1}^{k-1} |Q_j|$.

We now consider the reverse operation, increasing the capacity of bin $i$
from $C^-$ to $C^+$. Even before the capacity increase, the total capacity 
exceeds the total number of balls, so there is some first non-full bin $i''$
following $i$ in the clockwise order. By Invariant \ref{inv:forward} there
is no ball located after bin $i''$ that hash to bin $i''$ or before.

After increasing the capacity of bin $i$, we get $C^+-C^-$ new holes in
bin $i$, and then start trying to fill new holes recursively: if we
have a new hole in some bin $i'$, we look for a ball hashing to a
later bin that hash to bin $i'$ or an earlier bin. We know that we never
have to look past bin $i''$. Let $i'$ be the bin furthest from $i$
that loses a ball, creating a new hole that cannot be filled.

Let $i_0,\ldots,i_k$ be the bins from $i$ to $i'$.
For $j=0,\ldots,k-1$, let $Q_j$ be the set of balls 
that were used to fill a hole in bins $i_0,\ldots,i_j$ with a ball
starting from bins $i_{j+1},\ldots,i_k$.  In particular $Q_0$ is
the set of balls moved to fill the holes in $i$, and then 
$\sum_{i=0}^{k-1}|Q_j|$ is an upper bound on the number of moves to fill holes.

We now observe that $Q_0,\ldots,Q_{k-1}$ is a forwarding sequence
from a valid configuration after the capacity increase for bin $i$ and to 
a valid configuration before the capacity increase, that is, we get back
to the start if we for $j=1,\ldots,k-1$, forward the balls in $Q_j$ from
bin $i_j$ to $i_{j+1}$. As we saw above, this implies that
the balls and bins with capacites uniquely determines
$|Q_0|,\ldots,|Q_{k-1}|$, hence also our upper bound 
$\sum_{j=1}^{k-1} |Q_j|$ on the number of hole filling moves.
\end{proof}

The strength of Lemma \ref{lem:mov-cap-dec} is that when analyzing 
the number moves in connection with a capacity change, then we can look
at forwardings in any order we want, and get bounds that hold for any
sequence of moves following the description in Subsection \ref{sec:our-forward}.

\paragraph{Closing a bin}
When closing a bin $i$, we are going to lose its capacity $C\leq
\ceil{c\balls/\bins}$. Our first action is to increase
by one the capacities of $C$ other bins. Doing the increases first, we
make sure that the total capacity before every increase is always at
least $c\balls$. An increase is the inverse of the decrease that we
studied under full deletions, so the expected number of moves for each
increase is bounded by $O((\log c)/c^2)$. The expected total number of
moves resulting from all $C$ bin increases is thus
\[C\cdot O((\log c)/c^2)=O((\balls/\bins)(\log c)/c).\]
We are now ready to start closing $i$, having made sure that the total
capacity remains above $c\balls$. The closing have the same effect
as setting the capacity to zero, which may become overfull. We now
have to forward balls from overfull bins. 
By Lemma \ref{lem:mov-cap-dec},
when bounding the number of forwarding, we can perform them in any order
that is convinient for the analysis.

The first forwarding we do in our analysis is to transfer all the balls in $i$ to its
successor $i'$. The balls from $i$ either hashed directly to $i$, or
they were forwarded from the predecessor of $i$.  By Theorem
\ref{thm:exp-bins-c} (c), we expect to have $O(\balls/\bins)$ balls
that hashed directly to $i$ and $O((\balls/\bins)(\log c)/c)$ forwarded to
$i$ from its predecessor.  The balls from $i$ are now all moved to
$i'$. Bin $i$ is no longer active.

Now $i'$ may be overfull, and then we repeatedly forward balls from 
overfull bins until no overfull bins remain and Invariant \ref{inv:forward}
is restored. Let $i''$ be the last bin receiving a ball
from this forwarding. Then the total number of forwardings
is at most $C$ times the number of bins from $i'$ to the predecessor
of $i''$. All bins from $i'$ to $i''$ are
now full, but if we gave $i$ an extra capacity of one, 
then $i''$ would have received one less ball, hence
not be full. Applying by Theorem \ref{thm:exp-bins-c} (c),
to this situation, we know that the expected number
full bins starting from $i'$ to the predecessor of $i''$
is $O((\log c)/c^2)$, so the expected number
of forwardings is $C\,O((\log c)/c^2)\leq\ceil{c\balls/\bins}
O((\log c)/c^2)=O((\balls/\bins)(\log c)/c)$.
It might seem that we could here replace $C$ with its
expectation $O(\balls/\bins)$, but then we would be
multiplying two expectations that are not independent.

Summing up, we have proved that when closing $i$, besides the transfer to $i$
of $O(\balls/\bins)$ expected balls hashing directly to $i$, we have $O((\balls/\bins)(\log c)/c)$
ball moves. Here $\bins$ is the number of bins before the closing.

\paragraph{Opening a bin} When we open a bin it is like increasing
its capacity from zero to the desired capacity. Thanks to 
Lemma \ref{lem:mov-cap-dec}, we can bound the number of moves simply 
refering to the bound from above on the number of forwards used when we
above closed a bin, decreasing its capacity to zero. The subsequenct
unit capacity decreases are also symmetric to the unit capacity
increases we did when we closed a bin, so get the same overall bounds
for opening a bin as we did for closing a bin, that is,
besides the transfer to $i$
of $O(\balls/\bins')$ expected balls hashing directly to $i$, we have $O((\balls/\bins')(\log c)/c)$
ball moves. Here $\bins'$ is the number of bins after the opening of bin $i$,
but this only makes our bound better than if we used the
number $\bins=\bins'-1$ of balls before the opening.

This completes the proof of Theorem \ref{thm:moves-c} assuming the
correctness of Theorem \ref{thm:exp-bins-c}.

\subsection{Small capacities}
In the same way that we proved Theorem \ref{thm:moves-c} assuming 
Theorem \ref{thm:exp-bins-c}, we can prove Theorem \ref{thm:moves-eps} using
the following theorem.
\begin{theorem}\label{thm:exp-bins-eps} 
Consider a configuration with $\balls$ active balls and $\bins$ active bins 
and total capacity $(1+\bar\eps)\balls$ for some $\bar\eps\in (0,1]$. 
Suppose, moreover, that
each bin has capacity between $(1+\bar\eps)n/(2\bins)$ and $2(1+\bar\eps)\balls/\bins$. Then
starting from the hash location of a given passive ball or active or passive bin,
the expected number of consecutive full bins is $O(1/\bar \eps^2)$. 

Here we assume that balls and bins are hashed
independently, each using 5-independent hash functions or simple
tabulation hashing.
\end{theorem}
As in the analysis for Theorem \ref{thm:moves-c}, we will be operating
with an $\bar\eps\geq \eps$ such that the total capacity is exactly
$(1+\bar\eps)\balls$. Applying Theorem \ref{thm:exp-bins-eps} will then
yield a bound of $O(1/\bar\eps^2)=O(1/\eps^2)$. However, it
could be that while $\eps\leq 1$, we end up with $\bar\eps>1$. In
this case, we will instead apply Theorem \ref{thm:exp-bins-c},
and get a bound of $O((\log (1+\bar\eps))/(1+\bar\eps))=O(1)=O(1/\eps^2)$.

We note that Theorem \ref{thm:exp-bins-eps} is much simpler than
Theorem \ref{thm:exp-bins-c}. The point is that Theorem \ref{thm:exp-bins-c}
is used
to prove a loss by a factor $1+o(1)$ relative to simple consistent hashing
without capacities. For  Theorem \ref{thm:moves-eps}, the loss factor
is $O(1/\eps^2)$, which is much less delicate to deal with, e.g., for
the number of balls in a bin, instead of analyzing the expected number,
which is $O(\balls/\bins)$, we can just use the hard capacity, which is at most
$2(1+\bar\eps)\balls/\bins=O(\balls/\bins)$.

\section{Analysing expectations with large capacities}
In this section, we are going to prove Theorem \ref{thm:exp-bins-c}.

\subsection{Basic probability bounds}
We first briefly review the probability bounds we will use, and what
demands they put on the hash functions used. The general scenario
is that we have some bounded random variables $X_1,\ldots,X_n=O(1)$ and
$X=\sum_{i=1}^n X_i$. Let $\mu=\E[X]$. Assuming that the $X_i$ are 4-independent,
we have the fourth moment bound
\begin{equation}\label{eq:fourth}
\Pr[|X-\mu|\geq x]=O\left((\mu+\mu^2)/x^4\right).
\end{equation}
Deriving \req{eq:fourth} is standard (see, e.g., \cite{pagh07linprobe}). 
We will typically have the variable $X_i$ denoting that a ball (or a bin) hash
to a certain interval, which may be defined based on the hash of a certain
query ball $q$. If the hash function is 5-independent, the $X_i$ are
4-independent.

%

The fourth moment bound is very useful when $\mu\geq 1$. For smaller $\mu$,
we are typically looking for upper bounds $x$ on $X$, and there we
have much better bounds in the combinatorial case where $X_i$ indicates
if ball (or bin) $i$ lands in some interval. Suppose we know that the $X_i$ are
$a$-independent for some $a\leq x$.  
The probability that $a$ given balls land
in the interval is $p^a$, so the expected number of
such $a$-sets is ${\balls \choose a} p^a$. If we get $x$ or
more balls in the interval, then this includes at least ${x \choose a}$ such $a$-sets
in the interval. Thus, by Markov's inequality, with independence $a\leq x$,
\begin{equation}\label{eq:d>k}
\Pr[X\geq x]= (\mu^{\underline{a}}/a!)/(x^{\underline{a}}/a!)=O((\mu/x)^a)\textnormal{ for }a=O(1)
\end{equation}
where $x^{\underline a}$ is defined to be $x(x-1)\cdots(x-a+1)$.
We shall only use \req{eq:d>k} with $a\leq 3$. One advantage to this is that
our results will hold, not only with $5$-independent hashing, but also
with simple tabulation hashing. The point is that in  \cite{patrascu12charhash}
it is shown that while simple tabulation is only 3-independent, it does
satisfy the fourth moment bound \req{eq:fourth} even with a fixed hashing
of a given query ball. According to the experiments \cite{patrascu12charhash},
simple tabulation hashing is an order of magnitude  faster than 5-independent
hashing implemented via a degree-6 polynomial.

\subsection{Proof of Theorem \ref{thm:exp-bins-c}}
We now focus on the proof of Theorem \ref{thm:exp-bins-c}. Cheating a
bit, we will assume $\bar c\geq 64$. We will instead handle $\bar
c<64$ as a small capacity in Section \ref{sec:small-cap} where we
consider any $\bar c=1+\bar\eps=O(1)$.
\drop{ , and instead prove
  Theorem \ref{thm:exp-bins-eps} for any $\bar\eps=O(1)$, noting that
  for $\bar\eps,\bar c=\Theta(1)$, the bounds in both theorems are
  just constant.}  

Also, to increase readability, since the parameter $c\leq \bar c$ does not
appear in the analysis, we will just write $c$ instead of $\bar c$ below.
First we focus on proving Theorem \ref{thm:exp-bins-c} (a) and (b) as restated
in the following lemma.
\begin{lemma}\label{thm:cool-small}
Starting from the hash location of a given $q$, which is either a
passive ball or an active or passive bin, the expected number of
consecutive full bins is $O(1/c)$. If $q$ is a given bin with
capacity at least $2$, the expected number of consecutive full bins
is $O((\log c)/c^2)$.
\end{lemma}
\begin{proof}
We bound the expected number $d$ of consecutive full bins around $h(q)$. These
should include $q$ if $q$ is a bin, and the bin $q$ would hash to if $q$ is a passive
ball. If this bin is not full, then $d=0$ and we are done. Otherwise, let
$I=(a,b]\ni h(q)$ be the interval covered by these full bins, that is, $a$ is the
location of non-full bin preceding the first full bin, and $b$ is
the location of the last full bin. In our analysis, we assume that $h(q)$
is fixed before the hashing of any other balls and bins. 

We will study the event that $t^-\leq d<t^+$ and 
$\ell^-<|I|\leq \ell^+$. First we note that $d\geq t^-$ and $|I|\leq \ell^+$
imply the event:
\begin{description}
\item[$A(t^-,\ell^+)$:] Enough balls
to fill $t^-$ bins hash to $(h(q)-\ell^+,h(q)+\ell^+)$.
\end{description}
%

Also $d< t^+$ 
implies that at most $t^+-2$ bins land in
$I\setminus\{b\}=(a,b)$. We are discounting position $b$, because 
we could have additional bins hashing to $b$ that are not full. Note that
we could have $h(q)=b$. No matter how $I=(a,b]\ni h(q)$ is placed,
we have either $(a,b)\supseteq [h(q)-\ceil{\ell^-/2},h(q))$, or
$(a,b)\supseteq [h(q),h(q)+\ceil{\ell^-/2})$.
Thus $d< t^+$ and $|I|>\ell^-$ imply the event:
\begin{description}
\item[$B(t^+,\ell^-)$:] Either at most $t^+-2$ bins hash to 
$[h(q)-\ceil{\ell^-/2},h(q))$ or 
at most $t^+-2$ bins hash to $[h(q),h(q)+\ceil{\ell^-/2})$. 
\end{description}
Since balls and bins hash independently, after $h(q)$ has been fixed,
the events $A(t^-,\ell^+)$ and $B(t^+,\ell^-)$ are independent, 
so 
\[\Pr[t^-\leq d<t^+ \wedge \ell^-<|I|\leq \ell^+]\leq 
\Pr[A(t^-,\ell^+)]\Pr[B(t^+,\ell^-)].\]
For $i=0,1,\ldots$, let $t=2^i$, $t^-=t$, and $t^+=2t$. Moreover, define
$\ell(t)=\floor{8tr/\bins}$. Recall here that $r\geq \bins$ is the range we hash balls and bins to
to. We will bound $\Pr[t\leq d<2t]$ as follows.  
\begin{align}
\Pr[t\leq d<2t]\leq &\Pr[t\leq d<2t\wedge\, |I|\leq \ell(t)]\nonumber\\
&+
\sum_{j=0}^{\ceil{\log_2 c}}
\Pr[t\leq d<2t\wedge 2^j \ell(t)<|I|\leq 2^{j+1} \ell(t)]\nonumber\\
&+\Pr[t\leq d<2t\wedge c\ell(t)<|I|]\nonumber\\
\leq &\Pr[A(t,\ell(t))]\label{eq:compute}\\
&+
\sum_{j=0}^{\ceil{\log_2 c}}\Pr[A(t,2^{j+1} \ell(t))]\Pr[B(2t,2^j\ell(t))]\nonumber\\
&+\Pr[B(2t,c\ell(t))]\nonumber.
\end{align}
The main motivation for the definition of $\ell(t)$ is that with
$\ell^-=s\ell(t)$, $s\geq 1$, we can get a good fourth moment bound on
$\Pr[B(2t,\ell^-)]$. 

We consider the case where among all bins different from $q$ (which may be a ball or a bin), at most
$x=t^+-2=2(t-1)$ hash to $[h(q)-\ceil{\ell^-/2},h(q))$. We note that there are at least $\bins-1$ bins that are different from $q$.
We will
pay a factor 2 in probability to cover the equivalent case 
where they hash to $[h(q),h(q)+\ceil{\ell^-/2})$.  The
expected number of bins different from $q$ hashing to
$[h(q)-\ceil{\ell^-/2},h(q))$ is $\mu\geq (\bins-1)\ceil{\ell^-/2}/r$, and we
want this to be at least $2x=4(t-1)$. This is indeed the case with
$\ell^-\geq \ell(t)$ since $\ell(t)=\floor{8tr/\bins}\geq 8(t-1)r/(\bins-1)$.
Since $\bins\geq 2$, we have
\[\mu\geq (\bins-1)\ceil{\ell^-/2}/r\geq \bins s\ell(t)/(4r)\geq st\geq 1.\]
Applying the fourth moment bound \req{eq:fourth}, we now get 
\begin{equation}\label{eq:B4th}
\Pr[B(2t,s\ell(t))]=O((\mu+\mu^2)/(\mu-2(t-1))^4)=O(1/\mu^2)=(1/(st)^2).
\end{equation}
Next, we will develop different bounds for $\Pr[A(t,s\ell(t))]$ that are all
useful in different contexts. The interval $(h(b)-s\ell(t),h(b)+s\ell(t))$
has length $2s\ell(t)-1$, so the expected number of active balls hashing to it is
$\mu<2s\ell(t)n/r=O(st\balls/\bins)$.  However, to satisfy $\Pr[A(t,s\ell(t))]$,
we need enough balls to fill $t$ bins in that interval. 
Their total capacity is $x\geq tc\balls/(2\bins)$.
Suppose we also know that $x\geq a$ and that the balls hash $a$-independently.
Since we only assume 3-independence, $a\leq 3$. However, if $q$ is a ball,
with $h(q)$ fixed, we are left with 2-independence. 
Now, from \req{eq:d>k}, we get the bound
\begin{equation}\label{eq:A2}
\Pr[A(t,s\ell(t))]= O((\mu/x)^a)=O((s/c)^a).
\end{equation}
Thus, from \req{eq:compute}, we get 
\begin{align}
\Pr[t\leq d<2t]
\leq &\Pr[A(t,\ell(t))]+(\sum_{j=0}^{\ceil{\log_2 c}}\Pr[A(t,2^{j+1} \ell(t))]\Pr[B(2t,2^j\ell(t))])+\Pr[B(2t,c\ell(t))]\nonumber\\
=&O(1/c^a+(\sum_{j=0}^{\ceil{\log_2 c}}(2^{j+1}/c)^a(1/(2^jt)^2))+1/(ct)^2\label{eq:t<c}\\
=&\left\{\begin{array}{ll}
O(1/c)&\textnormal{if }a=1\\
O(1/c^2+(\log c)/(ct)^2)&\textnormal{if }a=2\\
O(1/c^3+1/(ct)^2)&\textnormal{if }a=3
\end{array}\right.\nonumber
\end{align}
When we start from the hash of an passive ball or bin, or an active bin whose
capacity might be only one, we can use $a=1$ for
$t=1$ bin and $a=2$ for $t\in[2,c]$. We now get
\begin{align*}
\E[d\cdot[d\leq c]]
&\leq \sum_{t=2^j,j=0,\ldots,\floor{\log_2c}} 2t\Pr[t\leq d<2t]\\
&=O(1/c)
+\sum_{t=2^j,\; j=1,\ldots,\floor{\log_2c}} t\,O(1/c^2+(\log c)/(ct)^2)\\
&=O(1/c).
\end{align*}
Above, for a Boolean expression $B$, we
have $[B]=1$ if $B$ is true, and $[B]=0$ if $B$ is false.

In the case where we start from a bin $q$ with capacity at least $2$, we can use $a=2$ for
$t=1$ bin. After the fixing of $h(q)$, the balls are still hashed 3-independently
and the capacity is at least $3$ with $t\geq 2$ bins, so for $t\in[2,c]$,
we can use $a=3$. Thus, when $q$ is a bin with capacity at least $2$, we get
\begin{align*}
\E[d\cdot[d\leq c]]&=\sum_{t=2^j,j=0,\ldots,\floor{\log_2c}} 2t \Pr[t\leq d<2t]\\
&=O((\log c)/c^2)
+\sum_{t=2^j,\; j=1,\ldots,\floor{\log_2c}} t\,O(1/c^3+1/(ct)^2)\\
&=O((\log c)/c^2).
\end{align*}
For $t\geq c$, we are instead going to use a fourth moment bound on
$\Pr[A(t,\ell(t)]$. To satisfy $A(t,\ell(t))$, we need 
$x\geq tc\balls/(2\bins)$ active balls to hash to $(h(q)-\ell(t),h(q)+\ell(t))$
but the expectation is only $\mu=(2\ell(t)-1)n/r\leq 2\floor{8tr/\bins}n/r\leq 16t\balls/\bins$.
Since $c\geq 64$, we have $x\geq 2\mu$. Moreover, $\mu\geq 16$ since $c\balls/\bins\geq 1$ and
$t\geq c$. Hence, by \req{eq:fourth}, we get
\begin{equation}\label{eq:A4th}
\Pr[A(t,\ell(t))]=O((\mu+\mu^2)/(t-\mu)^4)=O(\mu^2/t^4)=O((t\balls/\bins)^2/((tc\balls/\bins)^4))=
O(1/(tc)^2)
\end{equation}
For $t\geq c$, this replaces the $1/c^a$ term in \req{eq:t<c},
so we get
\begin{align*}
\Pr[t\leq d<2t]
\leq &\Pr[A(t,\ell(t)]+\sum_{j=0}^{\ceil{\log_2 c}}\Pr[A(t,2^{j+1} \ell(t))]\Pr[B(2t,2^j\ell(t))]+\Pr[B(2t,c\ell(t))]\\
=&O(1/(ct)^2)+\sum_{j=0}^{\ceil{\log_2 c}}(2^{j+1}/c)^a(1/(2^jt)^2+1/(ct)^2)\\
=&\left\{\begin{array}{ll}
O(\log c)/(ct)^2)&\textnormal{if }a=2\\
O(1/(ct)^2)&\textnormal{if }a=3
\end{array}\right.
\end{align*}
In the general case where we start from any passive ball, or active or
passive bin, we have $a=2$ for $t\geq 2$, so we get
\begin{align*}
\E[d\cdot[d\geq c]]&=\sum_{t=2^jc,j=0,\ldots,\infty} 2t \Pr[t\leq d<2t]\\
&=\sum_{t=2^jc,\;j=0,\ldots,\infty} 2t\,O((\log c)/(ct)^2)
=O(((\log c)/c^2)\textnormal,
\end{align*}
and then
\[\E[d]=\E[d\cdot[d<c]]+\E[d\cdot[d\geq c]]=O(1/c)+O(((\log c)/c^2)=O(1/c).\]
This completes the proof of the first statement of the theorem.

For the case where $q$ is a bin with capacity at least $2$, we have
$a=3$ for $t\geq 2$, and therefore
\begin{align*}
\E[d\cdot[d\geq c]]&=\sum_{t=2^jc,\;j=0,\ldots,\infty} 2t \Pr[t\leq d<2t]\\
&=O(\sum_{t=2^jc,\; j=0,\ldots,\infty} t\;O(1/(ct)^2)\\
&=O(1/c^2).
\end{align*}
Therefore, when $q$ is a bin with capacity at least $2$,
\[\E[d]=\E[d\cdot[d<c]]+\E[d\cdot[d\geq c]]=O(((\log c)/c^2)+O(1/c^2)=O(((\log c)/c^2).\]
This completes the proof of the theorem.
\end{proof}
We will now prove Theorem \ref{thm:exp-bins-c} (c) restated below.
\begin{lemma}
The expected number of balls hashing directly to any given active bin
$q$ is $O(\balls/\bins)$. The expected number of balls forwarded into $q$ from its
predecessor $q^-$ is $O(\balls/\bins(\log c)/c^2)$. Finally, if a bin
is not active, and its active successor $q^+$ is given an extra
capacity of one, then the expected number of of full bins starting
from $q^+$ is $O((\log c)/c^2)$.
\end{lemma}
\begin{proof}
For the first statement, we note that the expected number of balls hashing to
each location $h(q)-i$ is $\bins/r$ for any $0 \leq i \leq r$.
These are not added to $q$ if
some bin hashes to $[h(q)-i,h(q))$, which is an independent event because balls and bins
hash independently. The expected number of bins hashing to $[h(q)-i,h(q))$ is
$\mu=i(\bins-1)/r$. For $i\geq r/(\bins-1)$, we have $\mu\geq 1$, and
then, by \req{eq:fourth}, the probability of getting no bins in $[h(q)-i,h(q))$
is $O((\mu+\mu^2)/(\mu-0)^4)=O(1/\mu^2)=O((r/(\bins i))^2)$.
The expected number of balls hashing directly to $q$ is thus bounded by
\[n/r\cdot \left(\floor{r/(\bins-1)}+\sum_{i=\floor{r/(\bins-1)}+1}^\infty (r/(\bins i))^2\right)
=O(\balls/\bins).\] We also have to consider the probability that the preceding
bin $q^-$ forwarded balls to $q$. For this to happen, we would need
$q^-$ to be filled even if we increased its capacity by $1$, and then
$q^-$ would have capacity at least $2$. This is bounded by the
probability of having an interval $I\ni h(q)-1$ with $d\geq 1$
consecutive full bins including one with capacity at least $2$.  This
is what we analyzed in the proof of Lemma \ref{thm:cool-small}, so we
get $\Pr[d\geq 1]\leq \E[d]=O((\log c/c^2)$.  By the capacity
constraint, the maximal number of balls that can be forwarded to and
end in $q$ is $2c\balls/\bins$, so the expected number is 
\[O((\log c/c^2)2c\balls/\bins=O((\balls/\bins)(\log c)/c).\]
Next we ask for the expected number $d$ of full bins
starting from the active successor $q^+$ of a 
given passive bin $q$, when $q^+$ is given an extra
capacity of one. Again this implies that $q^+$ has capacity at least $2$,
and then the analysis from the proof of Lemma \ref{thm:cool-small}
implies that $\E[d]=O((\log c)/c^2)$.
\end{proof}

\section{Small capacities}\label{sec:small-cap}
In this section we will prove Theorem \ref{thm:exp-bins-eps} restated
below with $\eps$ instead of $\bar\eps$ and allowing any positive $\eps=O(1)$ 
instead
of just $\eps\in (0,1]$.
\begin{lemma}\label{lem:exp-bins-eps}
Consider a configuration with $\bins$ active balls and active $\bins$ bins 
and total capacity $(1+\eps)\balls$ for some $\eps=O(1)$. 
Suppose, moreover, that
each bin has capacity between $(1+\eps)n/(2\bins)$ and $2(1+\eps)\balls/\bins$. Then
starting from the hash location of a given passive ball or active or passive bin,
the expected number of consecutive full bins is $O(1/\eps^2)$. 

Here we assume that balls and bins are hashed
independently, each using 5-independent hash functions or simple
tabulation hashing.
\end{lemma}
We shall sometimes use a weighted count for number $Y_I$ of bins
hashing to an interval $I$, where the weight of a bin is its capacity
divided by the average capacity $(1+\eps)\balls/\bins$. These weights are
between $1/2$ and $2$.  The total capacity of the bins hashing into
$I$ is precisely $Y_I(1+\eps)\balls/\bins$.

Below we choose $\delta$ such that $(1+\delta)/(1-\delta)=(1+\eps)$.
For $\eps\leq 1$, we have $\delta\geq \eps/3$, but we also have for
any $\eps=O(1)$ that $\delta=\Omega(\eps)$. For $d\geq 6$, our goal is to show that
the probability of getting $d$ consecutive full bins is 
\begin{equation}
O(1/d^2\delta^4)=O(1/d^2\eps^4).
\end{equation}
It will then follow that the expected number of full bins is $O(1/\eps^2)$,
as required for Lemma \ref{lem:exp-bins-eps}.

Before doing the probabilistic analysis, we consider the combinatorial consequences
of having $d$ consecutive full bins.
\begin{lemma}\label{lem:search-conseq}  Let $p$ be the hash location
of an open or closed bin or an inactive ball from which
the number of successive full bins is $d\geq 2$. If we
do not have $d$ balls hashing directly to $p$, 
there there is an interval
$I$ containing $p$ of length at least $\ell=\floor{dr/(2\bins)}$, where either 
\begin{itemize} 
\item [(i)] the number balls $X_I$ landing in $I$
is $X_I\geq (1+\delta) \balls |I|/r$, or
\item[(ii)] the weighted number of bins $Y_I$ in $I$ is
$Y_I\leq (1-\delta) \bins |I|/r$.
\end{itemize}
\end{lemma}
\begin{proof} Let $I=(a,b]$ be the longest interval covered by consecutive full
bins around $p$. More precisely, $a$ is the hash location of the
non-full bin preceding the first full bin in the interval. We note
here that the first full bins hash strictly after $a$ because bins
hashing to the same location always get filled bottom-up. Since the preceding bin at $a$ is not full, bins in $I$ can only
be filled with balls hashing to $I$.

Since we do not have $d$ bins hashing to $p$, and we had $d$ full bins
starting from $p$, we must have $p<b$, so $p\in (a,b)$. Also, we have
at least $d$ full bins in $I$, so $I$ must contain at least
$d(1+\eps)n/(2\bins)$ balls.  Suppose $|I|\leq \ell$.  We can then
expand $I$ in either end to an interval $I^+\supseteq I$ of length
$\ell=\floor{dr/(2\bins)}$, and then $X_{I^+}\geq X_{I}\geq
(1+\eps)dn/(2\bins)\geq (1+\eps)n|I^+|/r> (1+\delta)n|I^+|/r$, so (i) is
satisfied for $I^+$.  Thus we may assume that $|I|\geq\ell+1$.

We now look at the interval $I^-=(a,b)$ which contains $p$ and is of length
at least $\ell$. All bins in $I^-$ are full.  Even though our last full bin
hashed to $b$, we have to exclude $b$ because  we might also have non-full bins
hashing to $b$.  

Since all bins hashing to $I^-$ are filled with balls hashing
to $I^-$, we have $X_{I^-}\geq (1+\eps) Y_{I^-} \balls/\bins$. We now reach a contradiction if (ii) and (i) are
false for $I^-$, for then 
\[X_{I^-}< (1+\delta) \balls |{I^-}|/r= (1+\eps)(1-\delta) \balls |{I^-}|/r<
(1+\eps) Y_{I^-} \balls/\bins.\]
\end{proof}

We note that applying  \req{eq:d>k} with $a=2$, we get that the probability
of $d$ bins hashing to $p$ is $O(1/d^2)=O(1/d^2\delta^4)$. Hence it
suffices to bound the probability of (i) and (ii).  We will do this
using a technical result from \cite{patrascu12charhash}.

To state the result from \cite{patrascu12charhash}, we first reconsider
the standard fourth moment bound \req{eq:fourth}. We 
are hashing $\balls$ balls into $[r]$. Let $\alpha=n/r$ be the density of the 
balls in $r$. Let
$I$ be an interval that may depend on the hash location of the inactive
query ball, and let $X_I$ be the number of active balls hashing to
$I$. Then $\E[X_I]=\alpha|I|$, and hence we can state \req{eq:fourth}
as 
\begin{equation}\label{eq:fourth-PT}
\Pr[|X_I-\alpha|I||\geq\delta\alpha|I|]=O\left(\frac{\alpha|I|+(\alpha|I|)^2}{(\delta\alpha|I|)^4}\right)
\end{equation}
Now \cite{patrascu12charhash} considered the general
event  $\cD_{\ell,\delta,p}$ that, for a given $p\in [r]$ 
that may depend on the hash value of an inactive query ball, 
there exists an interval $I$ containing $p$ of length at least $\ell$
such that the number of keys $X_I$ in $I$ deviates from the
mean $\alpha|I|$ by at least $\delta\alpha|I|$. As a perfect generalization of \req{eq:fourth}, \cite{patrascu12charhash} proved 
\begin{equation}\label{eq:fourth-lp}
\Pr[\cD_{\ell,\delta,p}]=
O\left(\frac{\alpha\ell+(\alpha\ell)^2}{(\delta\alpha\ell)^4}\right)
\end{equation}
The proof of \req{eq:fourth-lp} in  \cite{patrascu12charhash} only
assumes \req{eq:fourth-PT}, even though \req{eq:fourth-PT} only considers
a single interval whereas \req{eq:fourth-lp} covers all intervals around
a given point of some minimal length. 

The bound \req{eq:fourth-lp} immediately covers the probability that
there exists an interval $I\ni p$ of length at least $\ell$ which
satisfies (i).  As stated earlier, we may assume that $d\geq 6$, and
since $r\geq \bins$, we get that $\ell=\floor{dr/(2\bins)}\geq
dr/(3\bins)$. Moreover, the total capacity is $(1+\eps)\balls \leq O(\balls)$, and
the minimum bin capacity is $1$, so $\bins\leq O(\balls)$, and therefore
$\alpha\ell\geq dn/(3\bins)=\Omega(1)$. Thus the probability of (i)
is bounded by 
\[O\left((\alpha\ell+(\alpha\ell)^2)/(\delta\alpha\ell)^4\right)
=O(1/(\delta^4d^2))=O(1/(\delta^4d^2)).\]
The analysis above could be made tighter if $\balls \gg \bins$, but then the error probability
would just be dominated by that for (ii) on bins studied in the next part.

We now want to limit the probability of getting too few bins, that
is, for some $I\ni p$ containing $p$ of length at least $\ell\geq dr/(3\bins)$, 
\begin{itemize}
\item[(ii)] the weighted number of bins in $I$ is
$Y_I\leq (1-\delta) \bins |I|/r$ 
\end{itemize}
It makes no difference to \req{eq:fourth-lp} from \cite{patrascu12charhash}
that we apply it to bins instead of balls. We note 
the bin counts are weighed with weights below $2$. This is not
an issue because weights where considered in \cite{patrascu12charhash}, and the equations 
\req{eq:fourth} and  \req{eq:fourth-lp} both hold for weights bounded by
a constant. Another technical issue is if we in
Lemma \ref{lem:exp-bins-eps} start from an active query bin at
$p$. Then this bin is always included in the interval $I\ni p$.  It
should only have been included with probability $|I|/r$. A simple
solution is to only count the start bib with this probability, yielding a slightly
smaller value $Y'_I\leq Y_I$. We then instead bound the probability
that $Y'_I\leq (1-\delta) \bins |I|/r$, which is implied by (ii). Now the bin density is exactly $\alpha=m/r$
and $\alpha\ell\geq (\bins/r)dr/(3\bins)\geq d/3\geq 2$,
so the probability of (ii) is bounded by
\[O\left((\alpha\ell+(\alpha\ell)^2)/(\delta\alpha\ell)^4\right)
=O(d^2/(\delta d)^4)=O(1/(\delta^4d^2).\]
We have thus shown that $O(1/(\delta^4d^2))$ bounds the probability of both (i)
and (ii) in Lemma \ref{lem:search-conseq}, hence the
probability of getting $d\geq 1/\delta^2$ successive full bins.
It follows that the expected number of successive bins is
bounded by 
\[1/\delta^2+\sum_{d=\ceil{1/\delta^2}}^\infty O(1/(\delta^4d^2))
=O(1/\delta^2)=O(1/\eps^2)\]
This completes the proof of Lemma \ref{lem:exp-bins-eps}, and hence of
Theorem \ref{thm:exp-bins-eps}.

\section{Computing moves when the system is updated}\label{subsec:computemoves}
In this paper, we have not been concerned about computing the moves
that have to be done. In many applications this is a non-issue since
we can afford to recompute the configuration from scratch before and
after each update. These are also the applications where history
independence matters.  Our more efficient computation of moves is
based on a global simulation of the system in RAM. Our implementation
is tuned to yield the best theoretical bounds. 

We will now first describe an efficient implementation when balls are just
inserted in historical order, much like in the standard implementation
of linear probing. We will describe how to compute the moves
associated with an update in expected time proportional to the bounds
we gave above for the expected number of moves.

Given a ball $q$, we want to find the bin it hashes to in
expected constant time. This is the bin succeeding $q$ clockwise when
bins and balls are hashed to the cycle. For that we use an array
$B$ of size $t=\Theta(\bins)$. Entry $i$ is associated with the interval
$I_i=[i/t,(i+1)/t)$. In $B[i]$, we store the head of a list with the
  bins hashing to $I_i$ in sorted order. The list has expected
  constant size. When a ball $q$ comes, we compute the interval
  $I_i\ni h(q)$, and check $B[i]$ for its succeeding bin. If $L[i]$
  has no bin hashing after $q$, we go to $L_{i+1}$.  The expected time is
$O(1)$. As bins are
  inserted and deleted, we use standard doubling/halving techniques in the
back ground to make sure that $L$ always have $\Theta(\bins)$ entries.

For each bin, we store the balls landing in it in a hash table. Recall
here that when a ball is inserted, we first find the bin it hashes to,
but if it is full, we have to place it in the first non-full
succeeding bin. Such insertions are trivially implemented in time
proportional to the number of bins considered, which is exactly what
we bounded when we considered the number of moves.

We now turn to deletions. A deletion is essentially like a deletion in linear
probing. When we take out a ball $q$ from a bin $b$, we try to refill
the hole by looking at succeeding bins $b'$ for a ball $q'$ that hash
to or before $b$. We then move $q'$ to $b$, and recurse to fill the
hole left by $q'$ in $b'$.  The last hole created is the bin where $q$ would
land if $q$ was inserted last in the current configuration. We are
willing to pay time proportional to number of bins from the one $q$
hashed to, and to the one it would land in if inserted last.

To support efficient deletions, we let each bin $b$ have a forwarding
count for the number of balls
that hashed preceding $b$ but landed in a bin succeeding $b$. Also,
we divide the balls landing in $b$ according to which bins they originally
hashed to. More precisely, we maintain
a doubly-linked list with the bins that balls landing in $b$ hashed to, 
sorted in the same order as the bins appear on the cycle. With 
each of these preceding bins $b^-$, we maintain a list with the balls that
hashed to $b^-$ and landed in $b$. The total space for these lists is 
proportional to the number of balls landing in $b$, so the total space
remains linear. We assume that we can access the sorted bin list
from each end.

Now, if a deletion of a ball $q$ creates a hole in $b$, then this hole
can be filled if and only if the forwarding counter of $b$ is
non-zero. If so, we consider the succeeding bins $b'$ one by one. In
$b'$, we go to the beginning of the sorted bin list to find the first bin $b^-$ that a
ball landing in $b'$ hashed to. If $b^-$ equals or precedes $b$, we
can use any ball $q'$ from the list of $b^-$ in $b'$ and fill the hole
in $b$. Checking the bin $b'$ and possibly moving a ball takes
constant time. If a hole is created in $b'$ and $b'$ has non-zero
forwarding counter, we recurse.

Another issue is that we have is to locate a ball to be deleted in one of the
above lists. For that we employ an independent hash table for the
current balls point to their location. The hash table could be
implemented with linear probing which works in expected constant time
using the same 5-independent or simple tabulation hash function that
we used to map balls to the cycle
\cite{pagh07linprobe,patrascu12charhash}. Below it is assumed that we
always update the hash table with the location of the balls as we move
them around. An alternative solution would be that we searched the
ball starting from the bin it hashed to, and then only used a local
hash table for each bin.

When a ball $q$ is inserted, we have to update the above information.
If it hashes to a bin $b$ and lands in a bin $b'$, then we
have to increment the forwarding counter for all bins starting from
$b$ and going to the bin preceding $b'$---if $b=b'$, no forwarding
counter is incremented. Now inside $b'$, we have to go backwards in
the list of bins that balls landing in $b'$ hashed to, searching for $b$,
inserting $b$ if necessary. Next we add $q$ to the list of $b$ in $b'$. 
The bins considered inside $b'$ is 
a subset of the bins we passed from $b$ to $b'$, so the time bound is
not increased.

When bins are inserted and deleted, we implement the effect using insertion
and deletion of the affected balls as described above.

The most interesting challenge is that when a ball is inserted or
deleted, we have to change the capacity of $c$ bins. This
becomes hard when $c=\omega(1)$ where we only have $O(1)$
expected time available. 

As the system is described above, we let
the  lowest $\ceil{c\balls}-\bins\floor{c\balls/\bins}$ bins have large capacity
$\ceil{c\balls/\bins}$ while the rest have small capacity $\floor{c\balls/\bins}$. However,
we only have to guarantee that no bin has capacity above $\ceil{c\balls/\bins}$.

We now relax things a bit. With no history independence, we just maintain
the list of current bins in the order they were inserted (if one is deleted
and reinserted, we count it as new). The large bins form a prefix of this
list. We also relax the requirement
on the number of large bins to be $\ceil{c\balls}-\bins\floor{c\balls/\bins}\pm c$.
This means that when a ball is inserted or deleted, it doesn't have
to be exactly $c$ bins whose capacity we change.
Next we partition the list of current bins into groups of length at most $c$
with the property that the combined length of any two consecutive groups is
at least $c$. This partition is easily maintained in constant time
per bin update, including pointers so that we can from each bin can find the 
head of the group it belongs to. 

The changes to bin capacities are now done for one or two groups at the
time, telling only the group heads what the capacity is. Thus,
to check if a bin is full, we have to ask the group head about 
the capacity.

Now, when we change the capacity of a group, we find out if the
change in capacity means that some bins need changes to their
information. A bin requiring action is called critical. More precisely, a
large bin is critical if it is full. A small bin is critical if it is
full and its forwarding count is non-zero.  This implies that it would
be full if it became large.  By Lemma \ref{thm:cool-small}, a bin is
critical with probability $O((\log c)/c^2)$, so we only
expect $O((\log c)/c)=o(1)$ critical bins in a group. To identify critical
bins efficiently, we divide groups into subgroups of size at most $\sqrt c$,
using the same algorithm as we used for dividing into groups. For each
group and for each subgroup, we count the number of critical bins. Now
a group only has a critical bin with probability $O((\log c)/c)$, which
is our expected cost to check if there are critical bins. If
so, we check which subgroups have a positive count of critical bins.
If a subgroup has critical bins, we find them by checking all the bins
in the subgroup. Altogether, we pay $O(\sqrt c)$ time per critical
bin, so the expected time to identify the critical bins is
$O(\sqrt c (\log c)/c^2))=o(1)$.

We next consider an implementation with history independence. For
history independence, we first do a random permutation of the bin and
ball IDs. We can use the classic $\pi(x)=ax+b \bmod p$ where $p$ is a
random prime and $a$ and $b$ are uniformly random in $\mathbb Z_p$.
These permuted IDs are still history independent. The point now is
that if we have a set $X$ of $\Theta(k)$ permuted IDs, then we can
maintain order in this set by bucketing based on
$\floor{\pi(x)/k}$. The buckets are then in relative sorted order, and
we can easily maintain order within each bucket since each ID is
expected to end in bucket with $O(1)$ other IDs. As usual, we can use
doubling/halving if the set $X$ grows or decreases by a constant factor.  Now it is
easy to maintain an ordered list of permuted bin IDs where a prefix of
these are large bins, just like in our original description. Also, for
each bin $b$, we maintain the list of balls that hash to it, ordered
based on the permuted IDs.  We note that the balls in this list are
extracted based on the hashing to the unit cycle which is independent of
the permutation of the ball IDs. The balls hashing to $b$ are placed in $b$ and
succeeding bins based on the ordering of this list, that is, first
we fill $b$ with the balls with the lowest permuted IDs. In particular,
the list of balls hashing to $b$ but landing in bin $b'$ is just a segment
of the sorted list of balls hashing to $b$.


\section{High Probability Bounds}\label{sec:whp}
We are now going to present high probability bounds, that is, bounds
that holds with probability $1-1/\bins^\gamma$ for any desired
constant $\gamma$. The analysis assumes fully random hashing.
The 5-independent hashing that sufficed for our expected bounds
does not suffice for high probability. However, some of our high probability bounds
can also be obtained with simple tabulation \cite{patrascu12charhash}.

\begin{theorem}\label{thm:whp}
With balancing parameter $c=(1+\eps)=O(1)$, w.h.p., the maximal number
of bins considered in connection with a search is 
$O((\log n)/\eps^2)$.
This bound also holds for the number of balls that have to be moved when a ball is
added or removed. When a bin is added or removed, w.h.p.,
the number of balls to move is 
$O((\balls/\bins)(\log n)/\eps^2)$.

With balancing parameter $c=\omega(1)$, w.h.p., the maximal number
of bins considered in connection with a search is 
$1+O\left((\log \bins)/c\right)$ if  $\balls>\bins/2$, and 
$1+O\left(\frac{\log \bins}{(c\balls/\bins)\log(\bins/\balls)}\right)$
if  $\balls\leq \bins/2$.
This bound also holds for the number of balls that have to be moved when a ball is
added or removed. When a bin is added or removed, w.h.p.,
the number of balls to move is 
and $O((\balls/\bins)\log n)$ if $\balls>\bins/2$, and 
$O\left(\frac{\log \bins}{\log(\bins/\balls)}\right)$ if $\balls\leq \bins/2$.
\end{theorem}
If $\balls=\Omega(\bins)$, the above bounds are at most a log-factor
worse than the expected bounds. In the very lighlty loaded case where
$\balls=o(\bins)$, things have to get worse, e.g., even for simple
consistent hashing without forwarding, we expect some bins to have
$\Theta\left(\frac{\log \bins}{\log(\bins/\balls)}\right)$ balls,
matching our high probability bound for the number of balls to be
moved when a bin is closed.

The rest of this section is devoted to proving Theorem \ref{thm:whp}.
We will often study some run $R$ of consecutive bins full bins, preceeded by
an non-full bin at some position $x$. Let $y$ be the position of
the last bin in $R$. Because the preceeding bin is not full, 
all balls landing in the bins of $R$ must
hash to the interval $(x,y]$ covered by $R$. We note that the bin 
succeeding $R$ could be in position $y$, colliding with the last bin of $R$.

\paragraph{Large capacities}
We first study the case with super constant balancing 
parameter $c=\omega(1)$. What makes this case tricky to analyze is that
we in connection with ball updates have to make up to $\ceil{c}$ capacity
changes. Recall that bins have capacity $c\balls/\bins$ rounded up or
down. Also, recall that $c\balls/\bins\geq 1$.

For a given number $d$, we will bound the probability that a given bin
$q$ is in a maximal run $R$ of $d$ full bins. Based on $d$ we are
going to pick a threshold $\ell$. If the run covers an interval $I$ of
length at most $\ell$, then $I\subseteq (h(q)-\ell,h(q)+\ell]$, and
then we know that at least $dc\balls/(2\bins)$ hash to
$(h(q)-\ell,h(q)+\ell]$.  We call this the {\em ball event}. On the other
hand, if the $R$ covers interval of length bigger $\ell$, then, either
we have at most $d$ bins landing in $[h(q)-\ceil{\ell/2},h(q))$ or at
  most $d$ bins landing in $[h(q),h(q)+\ceil{\ell/2})$. We call this
the {\em left and right bin events}. 

It is convenient to reparameterize with  $s=\ell\bins/(dr)$, and we will
always have $s\geq 6$.

We will now bound the probability of the ball event. The expected
number of balls that hash to $(h(q)-\ell,h(q)+\ell]$ is
  $\mu_A=2\ell\balls/r= 2ds\balls/\bins$. Using the Poisson bound, the
  probability that at least $x_A=dc\balls/(2\bins)$ hash to this
  interval is bounded by
\begin{equation}\label{eq:balls-whp}
P_A=(e\mu_A/x_A)^{x_A}=((e2ds\balls/\bins)/(dc\balls/(2\bins))^{dc\balls/(2\bins)}=
((4es/c)^{dc\balls/(2\bins)}.
\end{equation}
Next we consider any one of the two bin events. In both cases, we want
at most $d$ balls to hash into an interval $I$ of length at
$\ell/2$. Ignoring the bin $q$, we expect at least
$\mu_B=(\ell/2)(\bins-1)/r=(ds/3)$ bins to land in $I$, and $(ds/3)\geq 2$. Therefore, by the Chernoff bounds, the probability of
each bin event is bounded by
\begin{equation}\label{eq:bins-whp}
P_B=\exp(-\mu_B/8)=\exp(-ds/24).
\end{equation}
Now, if $\bins=\Omega(\balls)$, we pick $s=c/(8e)$, which is bigger than $6$
since $c=\omega(1)$. With this parameter choice, the probability that
any of our events happen is bounded by $\exp(-\Omega(dc))$, which
hence also bounds the probability that bin $q$ is in a maximal run of 
length $d$. It follows that, w.h.p., the maximal run length is $O((\log n)/c)$.
In particular, it follows that for some sufficiently large $c=O(\log n)$,
w.h.p., there is no full bin.

If $\bins=o(\balls)$, we pick $s=24(c\balls/\bins)\ln(\bins/\balls)$.
Then 
\[P_A=(4e s/c)^{c\balls/(2\bins)}=(96e\balls/\bins\lg(\bins/\balls))^{c\balls/(2\bins)}=(\balls/\bins)^{\Omega(dc\balls/\bins)}\textnormal,\]
while 
\[P_B=\exp(-ds/24)=(\balls/\bins)^{dc\balls/\bins}\]
Thus we conclude that the probability that a bin $q$ is in a maximal run of 
length $d$ is bounded by $(\balls/\bins)^{\Omega(dc\balls/\bins)}$. It
follows that w.h.p., the maximal run length is $O((\log n)/(c\balls/\bins\log(\bins/\balls))$.

Above we studied the maximal run length, which also bounds the number of bins
we have to search if the first bin is full. It also bounds the number of
balls that have to be moved when a ball is added or removed, except
for the effect of changing $\ceil{c}$ bin capacities, each by 1. 
Recall here that adding and removing a ball has the same cost by
symmetry, so it suffices to consider a ball removal where capacities
are decreased by 1.

\paragraph{$\ceil{c}$ capacity decreases}
We will now study the number of moves in connection with
at most $\ceil{c}$ capacity decreases. Let $C$ be the
set of affected bins. The cost is bounded by the sum over the bins $p\in C$
of the length of the maximal run of full bins containing $p$. We already have
a bound on the maximal run length that we can multiply by $c$, but
we will show that, w.h.p., the sum of these run lengths is only
a constant factor bigger than the maximal run length. Since one bin
makes no difference here, we can assume that $c=\ceil{c}$.
The argument below assumes that $c\leq \sqrt{\balls}$. The
case $c>\sqrt{\bins}$ is quite extreme and can be handled with a different
argument.

Based on the previous analysis for a single run, we argue that any run
of full bins is contained in an interval of length at most
$O((\log\bins) r/\bins))$. To see this, recall that we gave high
probability bounds on (1) the maximum length $d$ of a run of full bins,
and (2) the length $\ell$ such that intervals of length $\ell$ have more
than $d$ bins. Here (1) and (2) implies that no run of full bin
can span an interval longer than $\ell$. For the case where
$\bins=\Omega(\balls)$, we had $s=c/(4e)$ and $d=O((\log n)/c)$, hence
\[\ell=sdr/\bins=(c/(4e)O((\log n)/c)r/\bins=O((\log\bins) r/\bins).\]
For the case, $\bins=\omega(\balls)$, we had
$s=24(c\balls/\bins)\ln(\bins/\balls)$ and 
$O((\log n)/(c\balls/\bins\log(\bins/\balls))$, hence
\[\ell=sdr/\bins=24((c\balls/\bins)\log (\bins/\balls))O((\log n)/((c\balls/\bins)\log (\bins/\balls)))r/\bins=O((\log\bins) r/\bins).\]
We fix the hashing of the bins in $C$ before we start hashing any of
the other balls and bins. We say that two bins from $C$ are
well-separated if they are at least $4\ell$ apart on the cycle. We
will partition $C$ into a constant number of sets, $C_1,..,C_{O(1)}$
so that bins in the same set are all well-separated. To do that we
partition the cycle into intervals of length between $4\ell$ and
$8\ell$. Consider one of these intervals $I$.  A bin from $C$ lands in
it with probability $O((\log \bins)/\bins)$ and
$c=|C|\leq\sqrt{\bins}$, so with high probability, we get only a
constant number of bins from $C$ in any $I$. To get our first set
$C_1$, we pick an arbitrary bin from every other interval, and
likewise we get $C_2$ picking a bin from each of the other
intervals. We have now picked a bin from every interval, so if we
repeat this a constant number of times, we get the desired
partitioning $C_1,..,C_{O(1)}$.

We now focus on any one of the well-separated $C_k$. Let
$c_k=|C_k|\leq c$. For each bin $p\in C_k$, we have a run of full bins,
and we want to bound the probability that the total length of these
runs is $d$. 

The run of full bins that includes $p\in C_k$ covers some interval
$I_p$ of length $\ell_p$, and let $\ell^*=\sum_{p\in C_k}\ell_p$.  As
when we studied the run of a single bin, we are going to define a
threshold length $\ell=O((\log n)r/$ and $s=\ell \bins/(dr)$.  Our
{\em combined ball event\/} is that there exists intervals $I_p$ of total
length $\ell^*=\ell$ with at least $dc\balls/(2\bins)$ hashing to
them. Our new {\em combined bin event\/} is that there exists intervals $I_p$
of total length $\ell^*=\ell$ with at most $d$ bins hashing inside of
them (not including the last position). If both events fail, then the
total run length cannot be $d$.

For each $p$, we will use a coarse overestimate $\ell^+_p$ of
$\ell_p$. We say that $\ell_p$ is on level $0$ if $\ell_p\leq
\ell/c_k$. Otherwise $\ell_p$ is on level $j$ where
$2^{j-1}\ell/c_k<\ell_p\leq 2^j\ell/c_k$. We set
$\ell^+_p=2^j\ell/c_k$. With these overestimates, we have
$I_p\subseteq I_p^+=(h(p)-\ell_p^+, h(p)+\ell_p^+]$ for all $p$.  The
bins in $C_k$ are well-separated so that there is no overlap between
any two $I_{p'}^+$.  Moreover, $\sum_{p\in C_k}\ell_p^+\leq
\ell+\sum_{p\in C_k}2\ell_p\leq 3\ell^*$, so $|\sum_{p\in C_k}
|I_p^+|\leq 6\ell^*$. Our combined ball event implies that $\ell^*\leq \ell$
and that $dc\balls/(2\bins)$ balls hash to $\bigcup_{p\in C_k}
I_p=\bigcup_{p\in C_k} I_p^+$.

Assume the $I_p^+$ were fixed. They have total length at most $6\ell$,
so the expected number of balls hashing to $\bigcup_{p\in C_k} I_p^+$
is at most $\balls (6\ell/r)=6ds\balls/\bins$.  The
probability that at least $dc\balls/(2\bins)$ balls hash to $\bigcup_{p\in C_k}
I_p^+$ is therefore at most
\begin{equation}\label{eq:comb-ball}
((e6ds\balls/\bins)/(dc\balls/(2\bins))^{dc\balls/(2\bins)}=
(12es/c)^{dc\balls/(2\bins)}.
\end{equation}
The important point here is that there are only $4^{c_k}$ possible
choices for the $I_i^+$. More precisely, for each level $j\geq 0$, we have a bit map telling
which of bins from level $\geq j$ that are on level $>j$. There are at most $c_k/2^{j-1}$ bins
on level $\geq j$, which is the size of the bit map for that level. The total number of
bits needed to describe all $I_p^+$ is thus at most $4c_k$, so 
there at most $2^{4c_k}$ combinations. This means that even
when we quantify over all possible combinations of
the intervals $I^+_p$, as long as their total length is $\ell$,
the total probability of the combined ball event is
bounded by 
\[2^{4c_k}((12es/c_k)^{dc\balls/(2\bins)}.\]
For the bins, we use the same idea, but starting from level $-1$ when
$\ell/(2c_k)\leq \ell_p<\ell/c_k$. We now use an underestimates
$\ell_p^-$ of $\ell_p$. We set $\ell_p^-=0$ if $\ell_p<\ell/(2c_k)$. Otherwise, if
$\ell_p$ is on level $j\geq -1$, we set $\ell_p^-=2^{j-2}\ell/c_k$. We
also have, for each bin a bit that tells if the interval $I_p$ is
mostly before or mostly after $h(q)$. If before, then $I_p$ contains
$I_p^-=[h(p)-\ell_p^-,h(p))$. Otherwise $I_p$ contains
$I_p^-=[h(p),h(p)+\ell_p^-)$. Including the before/after bitmaps, we
can completely describe all the $I_p^-$ using $6c_k$ bits, so there
are at most $2^{6c_k}$ combinations. Moreover, the total size of the
small $\ell_p<\ell/(2c_k)$ is at most $\ell/2$. For the $\ell_p\geq
\ell/(2c_k)$, we get $\ell_p^+\geq \ell_p/4$, so we conclude that
$\sum_{p\in C_k} \ell_p^-\geq \ell/8$. Assume that the $I_p^-$ are
fixed. The combined bin event implies that at most $d$ bins hash to
them. Considering only the $\bins-c$ bins outside $C$, we expect
at least $(\bins-c)(\ell/8)/r=(\bins-c)sd/(8\bins)\leq sd/10$ bins. With $s\geq
20$, the probability of the bin event is
$\exp(-(sd/10)/8)=\exp(-sd/80)$. Considering all possible configurations
of the bins $I_p^-$, we get that the probability of the combined bin event
is bounded by $2^{6c_k}\exp(-sd/80)$. 

If $\bins=\Omega(\balls)$, we pick $s=c_k/(24e)$ and conclude that the
probability of combined ball and bin events is bounded by 
$2^{O(c_k)}\exp(-\Omega(dc_k))=\exp(-\Omega(dc_k))$ for $d=\omega(1)$.
This bounds the probability that the sum of run lengths from the
bins in $C_k$ is $d$, so we conclude that, w.h.p., the total
run length is $O((\log n)/c)$. There are only $O(1)$ sets $C_k$ in $C$,
so we get the same bound on the total run length over all bins in $C$.

We would like to follow the same idea in the lightly loaded case where
$\bins=o(\balls)$. If we set $s=80 (c\balls/\bins)\ln(\bins/\balls)$,
we get that the probability of the combined ball and bin events is bounded by
$2^{O(c_k)}(\balls/\bins)^{\Omega(dc\balls/\bins)}$. However, this
time, the $2^{O(c_k)}$ may be so large. One issue is that $c_k$ may be
much larger than $d$.

We now note that if $d$ is total number of bins in the maximal runs of full 
bins involving bins in $C_k$, then at most $d$ bins from $C_k$ can be full.
Let us assume we guess the number $d'\leq d$ of full bins from $C_k$, and
let $C'\subseteq C_k$ be the subset of these bins. We can then apply the
the preceeding argument to $C'$ instead of $C_k$ with $d'$ instead of $c_k$.
Note that we are not conditioning on the bins from $C'$ all being full. 
We are merely studing the probability that $d$ is the total number of 
bins in the maximal runs of full bins involving bins in $C'$. To get
bounds for $C_k$, we have to sum over all $d'\leq d$ over
all ${c_k\choose d'}$ choices of $C'\subseteq C_k$. If $d'=O(1)$, we
are done, for then we involve only $O(1)$ maximal runs, so we may
assume that $d'=\omega(1)$.

When $C'$ is fixed, we only have $2^{6d'}$ choices for the coarse interals,
so for given $d'=\omega(1)$, the total number of configurations we have to consider
is 
\[{c_k\choose d'}2^{6d'}<({c_k}^{d'}/d'!)2^{6d'}<{c_k}^{d'}.\]
We will now sharpen our probability bounds for the combined bin and
ball events. Setting $s=160 (c\balls/\bins)\ln(\bins/\balls)$,
we get that the probability of the combined bin event is 
$\exp(-sd/80)=(\balls/\bins)^{2d c\balls/\bins}$. For $\bins=\omega(\balls)$,
this upper bound is maximized when $\balls/\bins$ is minimized, and
we know that $\balls/\bins\geq 1/c$, so 
$(\balls/\bins)^{2d c\balls/\bins)}\leq (1/c)^{2d}$ which is at least
the square of ${c_k}^{d'}$. Thus, covering all combinations, we conclude
that the combined bin event has probability at most $(\balls/\bins)^{d c\balls/\bins)}$.

When it comes to the ball bound, we need to exploit a subtle point
that we also used in our expected bound; namely that when we decrease
the capacity, it only has effect if the bin was full before the
decrease, and then the capacity was at least $2$. Our ball event
assuming that all bins in $C_k'$ are full thus implies that at least
$2d'$ balls hash to the intervals $I_p^+$ for $p\in C_k'$. We also
know that we need $dc\balls/(2\bins)$ to fill $d$ bins. Poisson bound
from \req{eq:comb-ball} can be refined to 
\begin{align}
P=(12es/c)^{\max\{dc\balls/(2\bins),2d'\}}&=
(12 e 160 (c\balls/\bins)\ln(\bins/\balls)/c)^{\max\{dc\balls/(2\bins),2d'\}}\nonumber\\
&=O((\balls/\bins)\log(\bins/\balls))^{\max\{dc\balls/(2\bins),2d'\}}.\label{eq:P}
\end{align}
For $\bins/\balls=\omega(1)$ and $\bins/\balls\leq c/4$, this bound grows
with $\bins/\balls$, so in this case, 
\[P\leq O(\log c/c)^{2d}\leq 1/c^{1.5d}.\]
On the other hand, for $c/4<\bins/\balls\leq c$, we use the $2d'$ bound,
and conclude that 
\[P\leq O(\log c/c)^{2d'}\leq 1/c^{1.5d'}.\]
Thus, including our factor $c^{d'}$ from the union over all configurations, 
we get the probability of the ball event is bounded by
\[c^{d'}P\leq P^{1/3}\leq 
(O(\balls/\bins)\log (\bins/\balls))^{\max\{dc\balls/(2\bins),2d'\}/3}
\leq (\balls/\bins)^{\Omega(d c\balls/(\bins))}.\] This completes the
proof that when $\bins=\omega(\bins)$, w.h.p., the total run length
over all bins in $C$ is $O((\log n)/((c\balls/\bins)\log
(\bins/\balls)))$.

\paragraph{Bin updates}
The last thing we need to analyze is the number of balls that may be
moved when a bin is added or removed. The two are symmetric, but the
easiest to think about is a bin removal. The direct effect is that the
balls in the bin are forwarded to the next been, but if it is full,
they are further forwarded. Then the total number of affected balls is
the bin capacity $\ceil{c\bins/\balls}$ times 1 plus the maximal number
of full bins in a run. We already have high probability bounds on the
maximal run length. In addition, we have to do $\ceil{c\bins/\balls}$ 
bin capacity changes. We already gave high probability bounds on the
number of balls that had to be moved with up to $c$ capacity changes,
so now we just have to multiply that number with $\ceil{\ceil{c\bins/\balls}/
c}\leq 1+\bins/\balls$. Thus, when a bin is added or removed, 
if $\balls=\Omega(\bins)$, then, w.h.p., the number of balls moved is bounded by
\[(1+c\balls/\bins)(1+O((\log \bins)/c))+(1+\bins/\balls)O((\log \bins)/c)=
O(c\balls/\bins+(\bins/\balls)(\log \bins)).\]
If $\balls=o(\bins)$, then, w.h.p., the number of balls moved is bounded by
\begin{align*}
(1+c\bins/\balls)&
\left(1+O\left(\frac{\log \bins}{(c\balls/\bins)\log(\bins/\balls)}\right)\right)+(1+\bins/\balls)O\left(\frac{\log \bins}{(c\balls/\bins)\log(\bins/\balls)}\right)).\\
&=O\left(c\bins/\balls+\left(\frac{\log \bins}{\log(\bins/\balls)}\right)\right).
\end{align*}
However, if we get no full bins, then there is no forwarding, so if a
bin is removed, then this only affects the ball that were in the bin
that are now transferred to the successor. Our high probability bound on
the maximal run length was, in fact, based on a bound on how many
balls we could have in an interval, large enough to have several bins,
with high probability, so multiplying that number with the capacity,
we get a high probability upper bound on how many balls hash directly
to any given bin. For $\balls=\Omega(\bins)$, the
bound is 
\[\ceil{c\balls/\bins}O\left((\log \bins)/c\right)=O\left((\balls/\bins)(\log \bins)\right).\]
If this bound is below $\floor{c\balls/\bins}$, then there are no
bins, so this bound replaces the previous $c\balls/\bins$ term.
Thus for $\balls=\Omega(\bins)$, we conclude, w.h.p., that if a bin added or removed, then the
total number of balls moved is $O\left((\log \bins)/c\right)$, as claimed
in Theorem \ref{thm:whp}.

Likewise, for $\balls=o(\bins)$, w.h.p., the maximal number of balls hashing
to a bin is 
\[\ceil{c\balls/\bins}O\left(\frac{\log \bins}{(c\balls/\bins)
\log(\bins/\balls)}\right)=\left(\frac{\log \bins}{\log(\bins/\balls)}\right).\]
Again this replaces the previous $c\balls/\bins$.
Thus for $\balls=o(\bins)$, 
we conclude, w.h.p., that if a bin added or removed, then the
total number of balls moved is $\left(\frac{\log \bins}{\log(\bins/\balls)}\right)$, as claimed
in Theorem \ref{thm:whp}. This completes the proof of Theorem \ref{thm:whp}
when $c=\omega(1)$.

\paragraph{Small capacities}
For the cases where $c=O(1)$, we get all he bounds from Theorem \ref{thm:whp}
follows if we prove, w.h.p., that the longest
run of full bins has length $O((\log n)/\eps^2)$. 

As in Section \ref{sec:small-cap}, we say that the weight of a bin
is its capacity divided by the $(1+\eps)\balls/\bins$. These weights
are between $1/2$ and $2$. Suppose there a run of full bins of
total weight at least $d$. We consider the shortest prefix $R$ of this
run of length at weight at least $d$. The bins have weight less than $d+2$.
Let $(x,y]$ be the interval covered by these bins. Here $x$ is
the location of the first bin and $y$ is the location of the
first bin after $R$ (for technical reasons, we made the tie breaking 
such that if a ball and bin landed in the same position, then the
ball proceeds the bin). We now have at least 
$d(1+\eps)\balls/\bins$ balls landing in $[x,y]$. We also know
that the total weight of the bins landing in $(x,y)$ is at most
$d+2$.

Let $\ell$ be any value. If $x-y\leq \ell$, then
we conclude that we have $d(1+\eps)\balls/\bins$ balls landing in some
interval of length $\ell+1$. On the other hand, if $x-y>\ell$,
then bins of weight at most $d+2$ landing in an interval of length
at most $\ell-1$.

We divide the cycle into at most $2\bins$ parts of length at most
$r/\bins\geq 1$. Inside an interval of length $\ell-1$, we can
have an interval of length at least $\ell^-=\ell-1-2r/\bins\geq \ell-3r/\bins$.
This interval can be chosen in at most $2\bins$ ways, and it 
that consists of a whole number of parts.

We will pick $\ell_d$ as a function of $d$.

Define $\delta$ such that $(1+\eps)=(1+\delta)^2/(1-\delta)^2$. Then
$\delta=\Omega(\eps)$. We set $\ell=d(r/\bins)/(1-\delta)^2$.

We will have $d\geq (\lg \bins)/\delta^2$. Then 
$\ell^-\geq (d+2)(r/\bins)/(1-\delta)$. It follows that the
expected weight of bins in such an interval $I^-$ is $(d+2)/(1-\delta)$.
Using standard Chernoff bounds, since each bin has weight below
2, the probability of this event is
\[\exp(-((d+2)/2)\delta^2/2))\leq \exp(-\Omega(d\eps^2)).\]
We can pick $d=O((\log n)/\eps^2)$ large enough that the
probability of an underfull interval of length $\ell^-$ is
$1/n^\gamma$.

We now consider the other error event: that there is an interval
of length $\ell^+$ with too many balls. The bad interval should 
consist of a minimal number of pieces so as to have length at
least $\ell+1$. It will have length at most $\ell^+=\ell+1+2r/\bins
\leq \ell+3r/\bins\leq \ell(1+\delta)$. The number of
balls expected in such an interval is 
$\balls \ell(1+\delta)/r=d(1+\delta)(\balls/bins)/(1-\delta)^2$. 
The error event is that we getting $d(1+\eps)\balls/\bins$, which is
$(1+\delta)$ times bigger than the expectation. By definition,
$(1+\eps)\balls/\bins\geq 1$, so by Chernoff bounds, the error
probability is 
\[\exp(-(d/(1-\delta))\delta^2/3))\leq \exp(-\Omega(d\eps^2)).\]
Again, we can drive the error probability down for the at most $2\bins$
possible combinations. This completes our proof of Theorem \ref{thm:whp}.

\paragraph{Tabulation hashing}
The bounds from Theorem \ref{thm:whp} with load balancing
$c=(1+\eps)=O(1)$ also hold with simple tabulation, which has somewhat
weaker Chernoff bounds \cite{patrascu12charhash}. These Chernoff
bounds are weaker in the exponent by a constant factor, and this only
costs us a constant factor in the maximal number of consecutive full
bins, and that suffices for $c=(1+\eps)=O(1)$. For $c=\omega(1)$, we
have to deal with bin capacity changes, which may not be sufficiently
independent with simple tabulation, but for $c=(\log n)^{1+o(1)}$, we
can still get the bounds of Theorem \ref{thm:whp} if we use twisted
tabulation \cite{PT13:twist}.


\section{Simulation Results}\label{sec:simulation}
To validate the consistency property of our hashing scheme which is
theoretically analysed in Theorems~\ref{thm:moves-c}, and
\ref{thm:moves-eps}, we present the following empirical results. We
generated thousands of instances, and tracked the number of ball
movements in each operation, and the distribution of bin sizes. We
picked the number of bins, $\bins$, the average balls per bin ratio $r
= \frac{m}{n}$, and $\eps$ as follows:

\begin{itemize}
\item
 $\bins \in \{10, 20, 40, 70, 100, 150, 200, 300, 450, 600, 800, 1000, 2000\}$.  
\item
$r = \frac{m}{n} \in \{0.5, 0.8, 1, 1.2, 1.5, 2, 3, 5, 10\}$
\item
$\eps \in \{0.05, 0.1, 0.2, 0.3, 0.4, 0.5, 0.6, 0.7, 0.8, 0.9, 1, 1.2, 1.5, 1.8, 2, 2.3, 2.5, 2.8, 3\}$
\end{itemize}

Figures~\ref{fig:bin-sizes} and \ref{fig:bin-sizes-consistent-hashing} show the distribution of bin loads for our algorithm and Consistent Hashing algorithm respectively. 
Figures~\ref{fig:ball-moves} and \ref{fig:bin-moves} depict the average number of movements in our algorithm for each value of $\eps$. 
We start with Figure~\ref{fig:bin-sizes} that shows the distribution of bin loads. 
The three plots represent three values of $\eps \in \{0.1, 0.3, 0.9\}$.
We expect the load of each bin to be at most $\lceil (1+\eps) \balls/\bins \rceil$. To unify the results of various simulations with different average loads ($\balls/\bins$), we divide the loads of bins by $\balls/\bins$, and sort the normalized bin loads in a decreasing order (breaking ties arbitrarily). The $y$ coordinate is the normalized load, and the $x$ coordinate shows the fractional rank of the bin in this order. For instance if a bin's load is greater than $35\%$ of other bins, its $x$ coordinate will be $35\%$. As we expect, no bin has normalized load more than $1+\eps$.  A significant fraction of bins have normalized load $1+\eps$, and the rest have normalized loads distributed smoothly in range $[0,1+\eps]$. The smaller the $\eps$ is, the more we expect bins to have normalized loads equal to $1+\eps$, and consequently having a more uniform load distribution.  

\begin{figure}[H]
\centering
\includegraphics[scale=0.3]{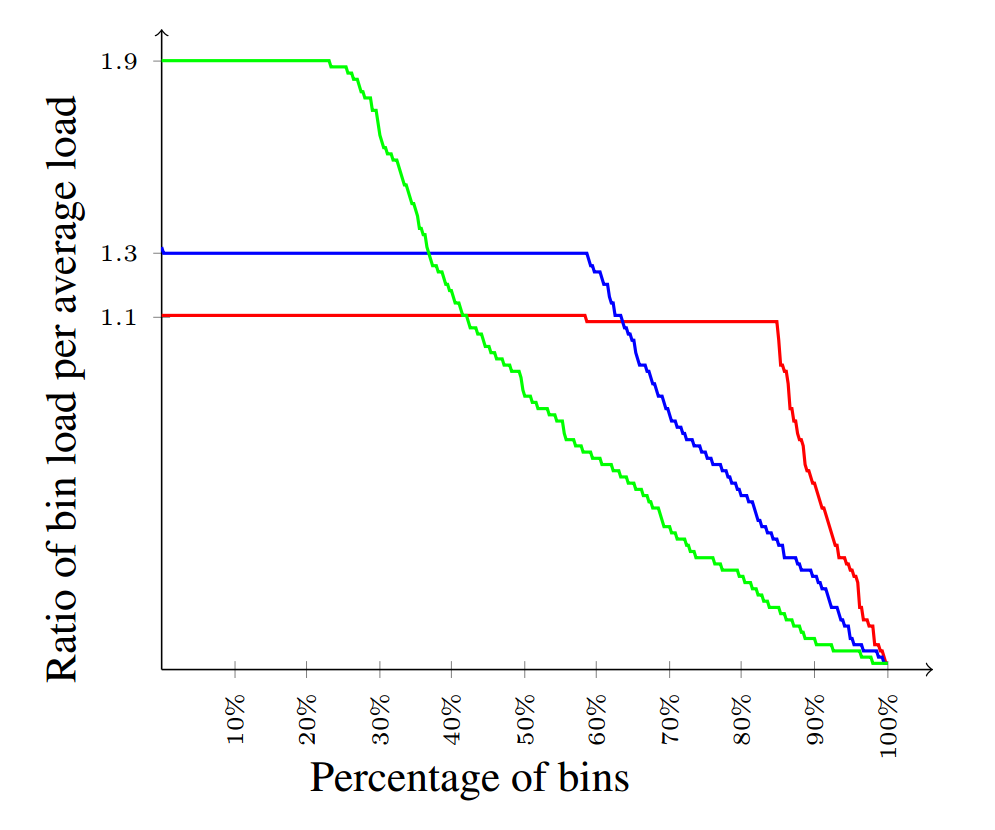}
\caption{Bin loads divided by average load $\balls/\bins$ for our algorithm.} \label{fig:bin-sizes}
\end{figure}

As shown above, the maximum load never exceeds $1+\eps$ times the average load in our algorithm. However, you can see below that there is no constant upper bound on maximum load for Consistent Hashing algorithm. We simulate consistent hashing with $\bins$ balls and $\bins$ bins for three different values of $\bins \in \{200, 1000, 8000\}$. As expected, the maximum load grows with $\bins$, it is expected to be around $\log(\bins) / \log\log(\bins)$. The percentage axis is rescaled to highlight the more interesting bin sizes. 

\begin{figure}[H]
\centering
\includegraphics[scale=0.3]{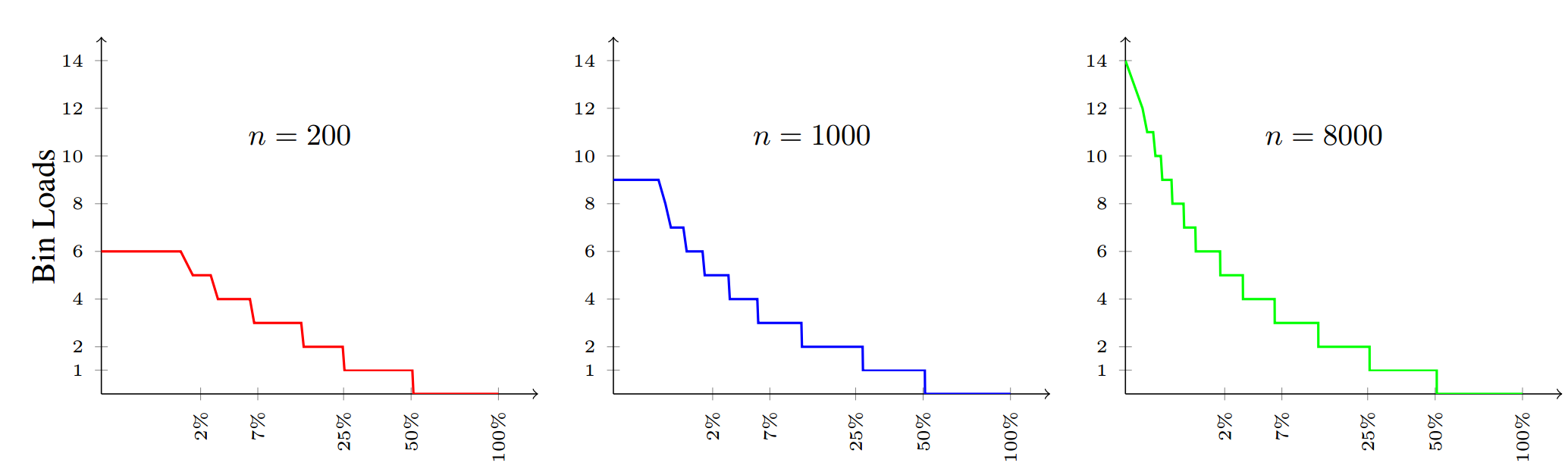}
\caption{Bin loads divided by average load for Consistent Hashing algorithm.} \label{fig:bin-sizes-consistent-hashing}
\end{figure}

Figure~\ref{fig:ball-moves} depicts the number of movements per ball operation. Theorems~\ref{thm:moves-c}, and \ref{thm:moves-eps} suggests that the  expected number of movements per ball operation is  $O(\log(c) / c)$, and $O(1 / \eps^2)$ respectively where $c$ is $1+\eps$. The solid red curve in Figure~\ref{fig:ball-moves} depict the average normalized ball movements in all simulations for each value of $\eps$. The bars show the standard deviation of these normalized movements. The dashed black line is the upper bound on these numbers of movements predicted by Theorems~\ref{thm:moves-c} and \ref{thm:moves-eps} with the following formula:

$$
f(\eps) =
\left\{
	\begin{array}{ll}
		2/\eps^2  & \mbox{if } \eps < 1 \\
		1 + \log(1+\eps) / (1+\eps) & \mbox{if } \eps \geq 1
	\end{array}
\right.
$$

\begin{figure}[H]
\centering
\includegraphics[scale=0.3]{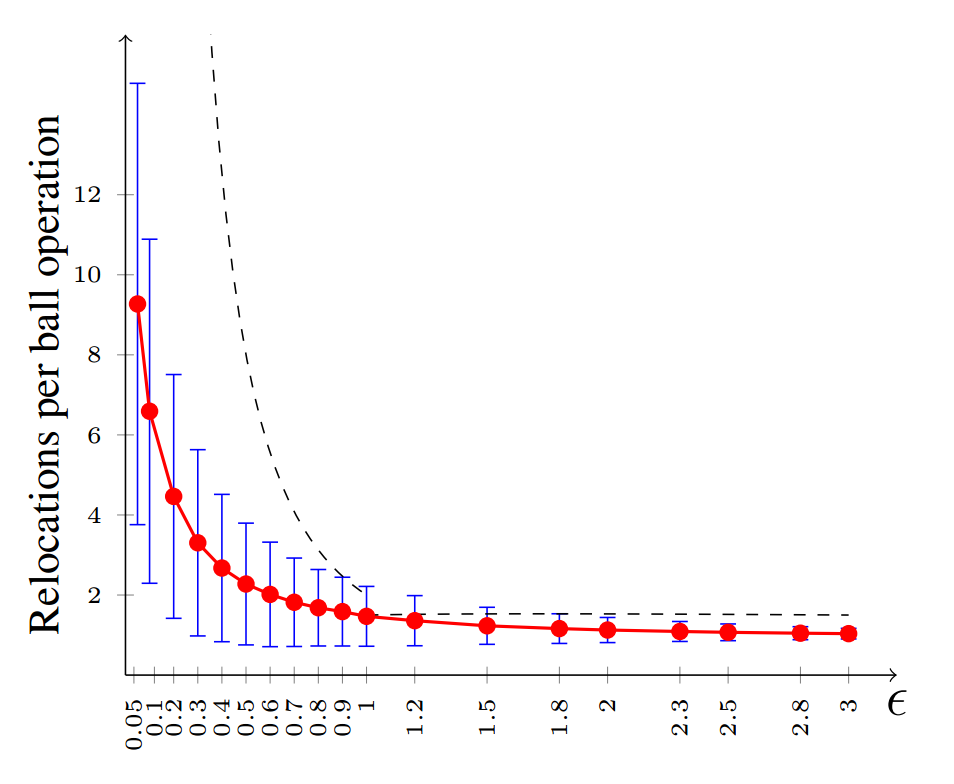}
\caption{Number of movements per ball operation for different values of $\eps$.} \label{fig:ball-moves}
\end{figure}

Our results predict that unlike the ball operations, the average movements per bin operation (insertion or deletion of a bin) is proportional to average density of bins $r = \frac{\balls}{\bins}$. 
Therefore we designate Figures~\ref{fig:bin-moves} and \ref{fig:ball-moves} to bin and ball operations respectively.
We start by elaborating on Figure~\ref{fig:bin-moves}.
Theorems~\ref{thm:moves-c} and \ref{thm:moves-eps} suggest that the average number of movements per bin operation is $O(r \log(c) / c)$, and $O(r / \eps^2)$ respectively where $c$ is $1+\eps$. To unify the results of all our simulations, we normalize the number of movements in bin operations with dividing them by $r$. The solid red curve in Figure~\ref{fig:bin-moves} depict the average normalized bin movements in all simulations for each value of $\eps$. The bars show the standard deviation of these normalized movements. The dashed black line is the upper bound function, $f(\eps)$ (defined above), on these numbers of movements predicted by Theorems~\ref{thm:moves-c} and \ref{thm:moves-eps}. Similarly Figure~\ref{fig:ball-moves} shows the relation between the number of movements and $\eps$ for ball insertions and deletions. These are the actual number of movements and are not normalized by the average density $r$. 

\begin{figure}[H]
\centering
\includegraphics[scale=0.3]{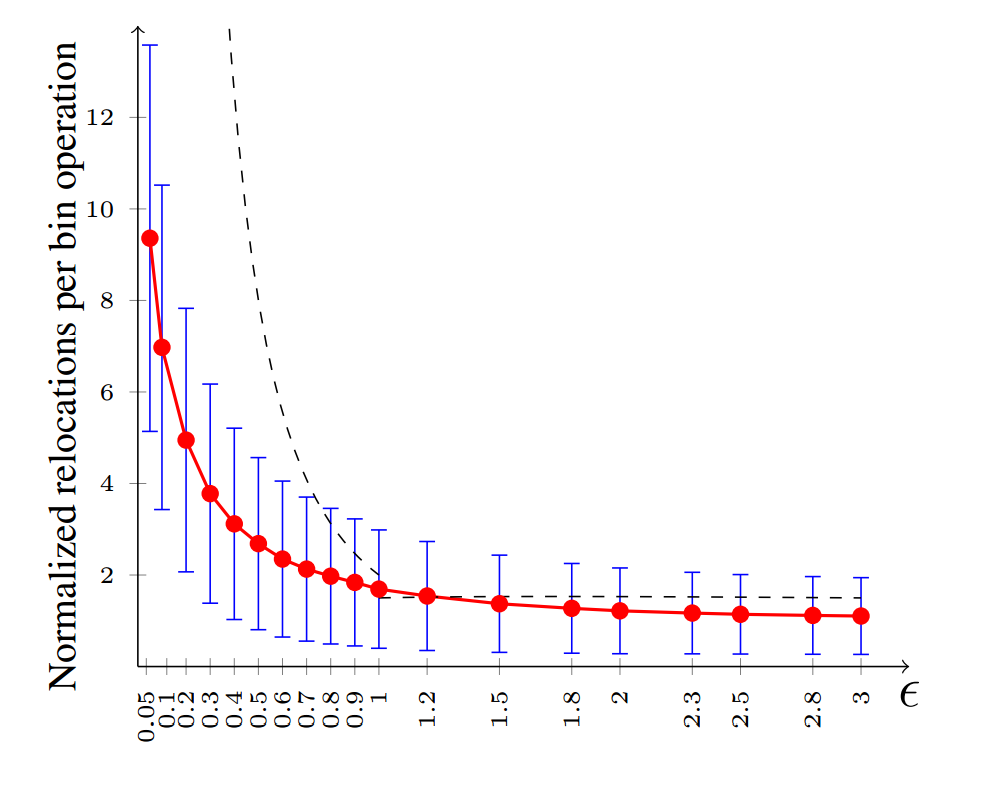} 
\caption{Normalized number of movements per bin operation for different values of $\eps$.}
\label{fig:bin-moves}
\end{figure}
\nocite{Lar88,DW07,Rod16,NSW08,FPSS05}

\paragraph{Acknowledgements}
We thank Andrew Rodland from Vimeo for putting our algorithmic idea into practice, 
and for letting use the plot from \cite{Rod16} presented in 
Figure~\ref{fig:vimeo} (with a changed background color to make printing easy).

{
\bibliographystyle{alpha} 
\bibliography{general}
}

\end{document}